\documentclass[11pt]{article}
\usepackage[margin=1in]{geometry}

\usepackage{enumitem}
\usepackage[utf8]{inputenc}
\usepackage{amsmath,amssymb,tabu,amsthm}
\usepackage{color}
\usepackage{calc}
\usepackage{tikz}
\usetikzlibrary{positioning,arrows,decorations.pathreplacing,shapes}
\usetikzlibrary{calc}
\usepackage{multirow}
\usepackage{boxedminipage}
\usepackage{xifthen}
\usepackage{tabularx}

\usepackage[round]{natbib}

\usepackage{mathpazo}

\usepackage{algorithmic}
\usepackage[linesnumbered,vlined,ruled]{algorithm2e}
\usepackage{hyperref}
\newcommand{\footremember}[2]{%
    \footnote{#2}
    \newcounter{#1}
    \setcounter{#1}{\value{footnote}}%
}

\newlength{\bibitemsep}
\newlength{\bibparskip}
\let\oldthebibliography\thebibliography
\renewcommand\thebibliography[1]{%
	\oldthebibliography{#1}%
	\setlength{\parskip}{\bibitemsep}%
	\setlength{\itemsep}{\bibparskip}%
}

\newcommand{\R}{\mathbf{R}}

\newcommand{\N}{\mathbf{N}}

\newcommand{\x}{\mathbf{x}}
\newcommand{\y}{\mathbf{y}}

\newcommand{\ALG}{ALG}
\newcommand{\OPT}{OPT}
\newcommand{\robpol}{\textbf{Rob-Pol}}

\newcommand{\tamapol}{\textbf{Tama-Pol}}
\newcommand{\utilpol}{\textbf{Util-Pol}}

\newcommand{\SPUA}{\textbf{SP-UA}}

\DeclareMathOperator{\E}{\mathbf{E}}
\DeclareMathOperator{\Prob}{\mathbf{P}}

\newtheorem*{claim*}{Claim}

\newtheorem{proposition}{Proposition}

\newtheorem{theorem}{Theorem}

\newtheorem{lemma}{Lemma}
\newtheorem{remark}{Remark}

\newlength{\RoundedBoxWidth}
\newsavebox{\GrayRoundedBox}
\newenvironment{GrayBox}[1]%
   {\setlength{\RoundedBoxWidth}{.93\textwidth}
    \def\boxheading{#1}
    \begin{lrbox}{\GrayRoundedBox}
       \begin{minipage}{\RoundedBoxWidth}}%
   {   \end{minipage}
    \end{lrbox}
    \begin{center}
    \begin{tikzpicture}%
       \node(Text)[draw=black!80,fill=white,rounded corners,%
             inner sep=2ex,text width=\RoundedBoxWidth]%
             {\usebox{\GrayRoundedBox}};
        \coordinate(x) at (current bounding box.north west);
        \node [draw=white,rectangle,inner sep=3pt,anchor=north west,fill=white]
        at ($(x)+(6pt,.75em)$) {\boxheading};
    \end{tikzpicture}
    \end{center}}

\newenvironment{defproblemx}[2][]{\noindent\ignorespaces%
                                \FrameSep=6pt%
                                \parindent=0pt%
                \vspace*{-1.5em}
                \ifthenelse{\isempty{#1}}{%
                  \begin{GrayBox}{\textsc{#2}}%
                }{%
                  \begin{GrayBox}{\textsc{#2} parameterized by~{#1}}%
                }
                \newcommand\Prob{Problem:}%
                \begin{tabular*}{\textwidth}{@{\hspace{.1em}} >{\itshape} p{1.8cm} p{0.8\textwidth} @{}}%
            }{
                \end{tabular*}%
                \end{GrayBox}%
                \ignorespacesafterend
            }

\usepackage{rotating}
\usepackage{framed}

\newcommand{\modif}[1]{\textcolor{black}{#1}}

\title{Robust Online Selection with Uncertain Offer Acceptance}

\author{
	Sebastian Perez-Salazar\footremember{1}{Rice University, \texttt{sperez@rice.edu}}
	\and Mohit Singh\footremember{2}{Georgia Institute of Technology, \texttt{mohit.singh@isye.gatech.edu}}
	\and Alejandro Toriello\footremember{3}{Georgia Institute of Technology, \texttt{atoriello@isye.gatech.edu}}
}

\begin{document}

\maketitle

\vspace{-0.5cm}
\begin{abstract}

Online advertising has motivated interest in online selection problems. 
Displaying ads to the right users benefits both the platform (e.g., via pay-per-click) and the advertisers (by increasing their reach). In practice, not all users click on displayed ads, while the platform's algorithm may miss the users most disposed to do so. This mismatch decreases the platform's revenue and the advertiser's chances to reach the right customers. 
With this motivation, we propose a secretary problem where a candidate may or may not accept an offer according to a known probability $p$. Because we do not know the top candidate willing to accept an offer, the goal is to maximize a robust objective defined as the minimum over integers $k$ of the probability of choosing one of the top $k$ candidates, given that one of these candidates will accept an offer. Using Markov decision process theory, we derive a linear program for this max-min objective whose solution encodes an optimal policy. The derivation may be of independent interest, as it is generalizable and can be used to obtain linear programs for many online selection models. We further relax this linear program into an infinite counterpart, which we use to provide bounds for the objective and closed-form policies. For $p \geq p^* \approx 0.6$, an optimal policy is a simple threshold rule that observes the first $p^{1/(1-p)}$ fraction of candidates and subsequently makes offers to the best candidate observed so far. 

\end{abstract}

\section{Introduction}

The growth of online platforms has spurred renewed interest in online selection problems, auctions and stopping problems \citep{edelman2007internet,lucier2017economic,devanur2009adwords,alaei2012online,mehta2014online}. Online advertising has particularly benefited from developments in these areas. As an example, in 2005 Google reported about \$6 billion in revenue from advertising, roughly 98\% of the company's total revenue at that time; in 2020, Google's revenue from advertising grew to almost \$147 billion. Targeting users is crucial for the success of online advertising. Studies suggest that targeted campaigns can double \emph{click-through rates} in ads \citep{farahat2012effective} despite the fact that internet users have acquired skills to navigate the web while ignoring ads \citep{cho2004people,dreze2003internet}. Therefore, it is natural to expect that not every displayed ad will be clicked on by a user, even if the user likes the product on the ad, whereas the platform and advertiser's revenue depend on this event \citep{pujol2015annoyed}. An ignored ad misses the opportunity of being displayed to another user willing to click on it and decreases the return on investment for the advertiser, especially in cases where the platform uses methods like pay-for-impression to charge the advertisers. At the same time, the ignored ad uses the space of another, possibly more suitable ad for that user. In this work, we take the perspective of a single ad, and we aim to understand the right time to begin displaying the ad to users as a function of the ad's probability of being clicked.

We model the interaction between the platform and the users using a general online selection problem. We refer to it as the \emph{secretary problem with uncertain acceptance} ($\SPUA$ for short). Using the terminology of \emph{candidate} and \emph{decision maker}, the general interaction is as follows: 
\begin{enumerate}[leftmargin=*]
	\item Similar to other secretary problems, a finite sequence of candidates of known length arrives online, in a random order. In our motivating application, candidates represent platform users.
	
	\item Upon an arrival, the decision maker (DM) is able to assess the quality of a candidate compared to previously observed candidates and has to irrevocably decide whether to extend an offer to the candidate or move on to the next candidate. This captures the online dilemma the platform faces: the decision of displaying an ad to a user is based solely on information obtained up to this point.
	
	\item When the DM extends an offer, the candidate accepts with a known probability $ p \in (0, 1] $, in which case the process ends, or turns down the offer, in which case the DM moves on to the next candidate. This models the users, who can click on the ad or ignore it. 
	
	\item The process continues until either a candidate accepts an offer 
or the DM has no more candidates to assess.
\end{enumerate}


A DM that knows in advance that at least one of the top $k$ candidates is willing to accept the offer would like to maximize the probability of making an offer to one of these candidates. In reality, the DM does not know $k$; hence, the best she can do is maximize the minimum of all these scenario-based probabilities. We call the minimum of these scenario-based probabilities the \emph{robust ratio} and our max-min objective the \emph{optimal robust ratio} (see Subsection~\ref{subsec:model_benchmark} for a formal description). Suppose that the DM implements a policy that guarantees a robust ratio $\gamma\in (0,1]$. 
This implies the DM will succeed with probability at least $\gamma$ in obtaining a top $k$ candidate, in any scenario where a top $k$ candidate is willing to accept the DM's offer. This is an ex-ante guarantee when the DM knows the odds for each possible scenario, but the policy is independent of $k$ and offers the same guarantee for any of these scenarios.
Moreover, if the DM can assign a numerical valuation to the candidates, a policy with robust ratio $\gamma$ can guarantee a factor at least $\gamma$ of the optimal offline value. \citet{tamaki1991secretary} also studies the $\SPUA$ and considers the objective of maximizing the probability of selecting the best candidate willing to accept the offer. Applying \citeauthor{tamaki1991secretary}'s policy to value settings can also guarantee an approximation factor of the optimal offline cost; however, the policy with the optimal robust ratio attains the largest approximation factor of the optimal offline value among rank-based policies (see Proposition~\ref{prop:utility_characterization}).

$\SPUA$ captures the inherent unpredictability in online selection, as other secretary problems do, but also the uncertainty introduced by the posibility of candidates turning down offers. $\SPUA$ is broadly applicable; the following are additional concrete examples.
\begin{description}[leftmargin=*]
	
	\item[Data-driven selection problems] When selling an item in an auction, buyers' valuations are typically unknown beforehand. Assuming valuations follow a common distribution, the aim is to sell the item at the highest price possible; learning information about the distribution is crucial for this purpose. In particular auction settings, the auctioneer may be able to sequentially observe the valuations of potential buyers, and can decide in an online manner whether to sell the item or continue observing valuations. Specifically, the auctioneer decides to consider the valuation of a customer with probability $p$ and otherwise the auctioneer moves on to see the next buyer's valuation. The auctioneer's actions can be interpreted as an \emph{exploration-exploitation} process, which is often found in bandit problems and online learning \citep{cesa2006prediction,hazan2019introduction,freund1999adaptive}. 
	This setting is also closely related to data-driven online selection and the prophet inequality problem \citep{campbell1981choosing,kaplan2020competitive,kertz1986stop}; some of our results also apply in these models (see Section \ref{sec:upper_bound}).
	
	\item[Human resource management] As its name suggests, the original motivation for the secretary problem is in hiring for a job vacancy. Screening resumes can be a time-consuming task that shifts resources away from the day-to-day job in Human Resources. Since the advent of the internet, several elements of the hiring process can be partially or completely automated; for example, multiple vendors offer automated resume screening \citep{raghavan2019mitigating}, and machine learning algorithms can score and rank job applicants according to different criteria. Of course, a highly ranked applicant may nevertheless turn down a job offer. Although we consider the rank of a candidate as an absolute metric of their capacities, in reality, resume screening may suffer from different sources of bias \citep{salem2019closing}, but addressing this goes beyond our scope. See also \citep{smith1975secretary,tamaki1991secretary,vanderbei2012postdoc} for classical treatments. Similar applications include apartment hunting \citep{bruss1987optimal,cowan1979optimal,presman1973best}, among others.
\end{description}


\subsection{Our Contributions}

(1) We propose a framework and a robust metric to understand the interaction between a DM and competing candidates, when candidates can reject the DM's offer. (2) We state a linear program (LP) that computes the optimal robust ratio and the best strategy. We provide a general methodology to derive our LP, and this technique is generalizable to other online selection problems. (3) We provide bounds for the optimal robust ratio as a function of the probability of acceptance $p\in (0,1]$. (4) We present a family of policies based on simple threshold rules; in particular, for $p\geq p^* \approx 0.594$, the optimal strategy is a simple threshold rule that skips the first $p^{1/(1-p)}$ fraction of candidates and then makes offers to the best candidate observed so far. We remark that as $p\to 1$ we recover the guarantees of the standard secretary problem and its optimal threshold strategy. (5) Finally, for the setting where candidates also have non-negative numerical values, we show that our solution is the optimal approximation among rank-based algorithms of the optimal offline value, where the benchmark knows the top candidate willing to accept the offer. The optimal approximation factor equals the optimal robust ratio.

\subsection{Problem formulation}\label{subsec:model_benchmark}

A problem instance is given by a fixed probability $p\in(0,1]$ and the number of candidates $n$. These are ranked by a total order, $1 \prec 2 \prec \cdots \prec n $, with $1$ being the best or highest-ranked candidate. The candidate sequence is given by a random permutation $ \pi = (R_1, \dotsc, R_n) $ of $[n]\doteq\{1,2,\ldots,n \}$, where any permutation is equally likely. At time $t$, the DM observes the partial rank $r_t \in [t]$ of the $t$-th candidate in the sequence compared to the previous $t-1$ candidates. The DM either makes an offer to the $t$-th candidate or moves on to the next candidate, without being able to make an offer to the $t$-th candidate ever again. If the $t$-th candidate receives an offer from the DM, she accepts the offer with probability $p$, in which case the process ends. Otherwise, if the candidate refuses the offer (with probability $1-p$), the DM moves on to the next candidate and repeats the process until she has exhausted the sequence. A candidate with rank in $[k]$ is said to be a \emph{top $k$ candidate}. \modif{The goal is a policy that maximizes the probability of extending an offer to a highly ranked candidate that \emph{will accept the offer}. However, since the DM does not know which candidate will accept the offer, the DM would like to be robust against any possible scenario. To measure the quality of a policy $\mathcal{P}$, we use the \emph{robust ratio}
\begin{align}
	\gamma_{\mathcal{P}}=\gamma_{\mathcal{P}}(p) = \min_{k=1,\ldots,n} \frac{\Prob(\text{$\mathcal{P}$ selects a top $k$ candidate}, \text{ candidate accepts offer} )}{\Prob(\text{At least one of the top $k$ candidates accepts offer})}. \label{eq:robust_ratio_pol}
\end{align}
The $k$-th term in the minimization operator, $\gamma_{\mathcal{P},k}(p)$, is the probability that policy $\mathcal{P}$ successfully selects a top $k$ candidate \emph{given} that some top $k$ candidate will accept the offer. Then, the \emph{robust ratio} $\gamma_{\mathcal{{P}}}=\min_{k=1,\ldots,n} \gamma_{\mathcal{P},k}(p)$ captures the situation where policy $\mathcal{P}$ has the worst possible performance over all such scenarios. When every candidate accepts an offer with certainty, $ p = 1 $, the robust ratio $\gamma_\mathcal{P}$ equals the probability of selecting the highest ranked candidate, thus we recover the standard secretary problem and $\gamma_{\mathcal{P}}(1)\approx 1/e$ for the optimal policy $\mathcal{P}$. The goal is to find a policy that maximizes this robust ratio, $\gamma_n^*\doteq \sup_{\mathcal{P}} \gamma_{\mathcal{P}}$. We say the policy $\mathcal{P}$ is \emph{$\gamma$-robust} if $\gamma \leq \gamma_\mathcal{P}$.}
	
\modif{\paragraph{The Robust Ratio and Related Objectives} The $\SPUA$ has been studied before under different objectives. \cite{smith1975secretary} studied the $\SPUA$ with the objective of maximizing the probability of selecting the top candidate and having that candidate accept the offer. This is the unconditional version of $\gamma_{\mathcal{P},k}$ for $k=1$; however, the top candidate may not accept the offer, and the objective does not plan for this contingency. This is particularly inadequate when $ p $ is small, as in many of our motivating applications.} 
%

\modif{
\cite{tamaki1991secretary} instead studied the $\SPUA$ with the objective of maximizing the probability of choosing the top candidate willing to accept the offer. Despite being more realistic than \citet{smith1975secretary}, this objective is often overly selective and may not make an offer hoping to encounter a better candidate in the future. Our objective overcomes this selectiveness, and makes an offer to a candidate 
as long as their rank is high compared to other candidates willing to accept the offer. A further distinction between \citeauthor{tamaki1991secretary}'s objective and the robust ratio emerges when values are assigned to the candidates. In this case, a value-driven DM would like to maximize the value obtained from a candidate that accepts the offer. The robust ratio turns out to be the optimal approximation ratio of any rank-based algorithm in this setting (see Proposition~\ref{prop:utility_characterization}). In Section~\ref{sec:experiments}, we provide extensive numerical experiments for the value version of the problem. Our policy consistently yields better results for small acceptance probabilities, $ p < 0.2 $, demonstrating its effectiveness compared to \cite{tamaki1991secretary}.}


\subsection{Our technical contributions}\label{subsec:technical_contributions}

Recent works have studied secretary models using linear programming (LP) methods  \citep{buchbinder2014secretary,chan2014revealing,correa2020sample,dutting2020secretaries}. We also give an LP formulation that computes the best robust ratio and the optimal policy for our model. Whereas these recent approaches derive an LP formulation using ad-hoc arguments, our first contribution is to provide a general framework to obtain LP formulations that give optimal bounds and policies for different variants of the secretary problem. The framework is based on Markov decision process (MDP) theory \citep{altman1999constrained,puterman2014markov}. This is surprising since early literature on secretary problem used MDP techniques, e.g.\ \citet{dynkin1963optimum,lindley1961dynamic}, though typically not LP formulations. In that sense, our results connect the early algorithms based on MDP methods with the recent literature based on LP methods. Specifically, we provide a mechanical way to obtain an LP using a simple MDP formulation (Section \ref{sec:lp_finite_formulation}). Using this framework, we present a structural result that completely characterizes the space of policies for the $\SPUA$:
\begin{theorem}\label{thm:main_polyhedron}
	Any policy $\mathcal{P}$ for the $\SPUA$ can be represented as a vector in the set
	\[
	\textsc{Pol} = \left\{ (\x,\y) \geq 0   : x_{t,s} + y_{t,s} = \frac{1}{t}\sum_{\sigma=1}^{t-1} \left(y_{t-1,\sigma} + (1-p) x_{t-1,\sigma} \right), \forall t> 1, s\in [t], x_{1,1} + y_{1,1} = 1 \right\} .
	\]
	Conversely, any vector $(\x,\y)\in \textsc{Pol}$ represents a policy $\mathcal{P}$. The policy $\mathcal{P}$ makes an offer to the first candidate with probability $x_{1,1}$ and to the $t$-th candidate with probability ${t x_{t,s}}/ {\left(\sum_{\sigma=1}^{t-1}y_{t-1,\sigma} + (1-p) x_{t-1,\sigma}\right)}$ if the $t$-th candidate has partial rank $r_t=s$.
\end{theorem}
The variables $x_{t,s}$ represent the probability of reaching candidate $t$ and making an offer to that candidate when that candidate has partial rank $s\in [t]$. Likewise, variables $y_{t,s}$ represent the probability of reaching candidate $t$ and \modif{not making an offer} when this candidate's partial rank is $s\in [t]$. We note that although the use of LP formulations in MDP is a somewhat standard technique, see e.g.\ \citet{puterman2014markov}, the recent literature in secretary problems and related online selection models does not appear to make an explicit connection between LP's used in analysis and the underlying MDP formulation. 

Problems solved via MDP can typically be formulated as reward models, where each action taken by the DM generates some immediate reward. Objectives in classical secretary problems fit in this framework, as the reward (e.g.\ the probability of selecting the top candidate) depends only on the current state (the number $t$ of observed candidates so far and the current candidate's partial rank $r_t = s$), and on the DM's action (make an offer or not); see Section \ref{sec:subsec:warm_up} for an example. Our robust objective, however, cannot be easily written as a reward depending only on $r_t=s$. Thus, we split the analysis into two stages. In the first stage, we deal with the space of policies 
and formulate an MDP for our model with a generic utility function. 
The feasible region of this MDP's LP formulation corresponds to $\textsc{Pol}$ and is independent of the utility function chosen; therefore, it characterizes all possible policies for the $\SPUA$. In the second stage, we use the structural result in Theorem~\ref{thm:main_polyhedron} to obtain a linear program that finds the largest robust ratio.
\begin{theorem}\label{thm:main_lp}
	The best robust ratio $\gamma_n^*$ for the $\SPUA$ equals the optimal value of the linear program
	
	{\hfil $(LP)_{n,p}$ \hfil \begin{tabular}{rp{0.7\linewidth}}
		$\displaystyle\max_{\x\geq 0}$ & \quad \quad $\gamma$ \\
		\text{{\textnormal{s.t.}}} &  {\vspace{-1cm}\begin{align*}
				x_{t,s} & \leq \frac{1}{t} \left( 1 - p \sum_{\tau=1}^{t-1} \sum_{\sigma=1}^\tau x_{\tau,\sigma}  \right) & \forall t\in [n], s\in [t] \\
				\gamma&\leq \frac{p}{1-(1-p)^k} \sum_{t=1}^n \sum_{s=1}^t x_{t,s}\Prob(R_t\leq k \mid r_t=s)& \forall k\in [n] ,
		\end{align*}\vspace{-0.5cm}}
	\end{tabular}}

where $\Prob(R_t\leq k \mid r_t=s)= \sum_{i=s}^{k\wedge (n-t+s)} \binom{i-1}{s-1}\binom{n-i}{t-s}/\binom{n}{t}$ is the probability the $t$-th candidate is ranked in the top $k$ given that her partial rank is $s$.

Moreover, given an optimal solution $(\x^*,\gamma_n^*)$ of $(LP)_{n,p}$, the (randomized) policy $\mathcal{P}^*$ that at state $(t,s)$ makes an offer with probability ${t x_{t,s}^*}/{\left( 1 - p \sum_{\tau=1}^{t-1} \sum_{\sigma=1}^\tau x_{\tau,\sigma}^*  \right)}$ is $\gamma_n^*$-robust.
\end{theorem}

We show that $\gamma_\mathcal{P}$ can be written as the minimum of $n$ linear functions on the $\x$ variables in $\textsc{Pol}$, where these variables correspond to a policy's probability of making an offer in a given state. Thus our problem can be written as the maximum of a concave piecewise linear function over $\textsc{Pol}$, which we linearize with the variable $\gamma$. 
By projecting the feasible region onto the $(\x,\gamma)$ variables 
we obtain $(LP)_{n,p}$.

As a byproduct of our analysis via MDP, we show that $\gamma_n^*$ is non-increasing in $n$ for fixed $p\in (0,1]$  (Lemma~\ref{lem:non_increasing_gamma}), and thus $\lim_{n \to \infty} \gamma_n^*=\gamma_\infty^*$ exists. 
We show that this limit corresponds to the optimal value of an infinite version of $(LP)_{n,p}$ from Theorem~\ref{thm:main_lp}, where $n$ tends to infinity and we replace sums at time $t$ with integrals (see Section~\ref{sec:infinite_LP}). 
This allows us to show upper and lower bounds for $\gamma_n^*$ by analyzing $\gamma_\infty^*$. Our first result in this vein gives upper bounds on $\gamma_\infty^*$.
\begin{theorem}\label{thm:main_upper_bound}
	For any $p\in (0,1]$, $\gamma_\infty^*(p) \leq \min \left\{p^{{p}/{(1-p)}}, 1/\beta \right\}$, where $1/\beta\approx 0.745$ and $\beta$ is the (unique) solution of the equation $\int_0^1 {(y(1-\log y) + \beta-1)^{-1} }\mathrm{d}y=1$.
\end{theorem}

To show $\gamma_\infty^*\leq p^{p/(1-p)}$, we relax all constraints in the robust ratio except $k=1$. This becomes the problem of maximizing the probability of hiring the top candidate, which has a known asymptotic solution of $p^{1/(1-p)}$ \citep{smith1975secretary}. For $\gamma_\infty^*(p)\leq 1/\beta$, we show that any $\gamma$-robust ordinal algorithm can be used to construct an algorithm for i.i.d.\ prophet inequality problems with a multiplicative loss of $(1+o(1))\gamma$ and an additional $o(1)$ additive error. Using a slight modification of the impossibility result by~\cite{hill1982comparisons} for the i.i.d.\ prophet inequality, we conclude that $\gamma_\infty^*$ cannot be larger than $1/\beta$.

By constructing solutions of the infinite LP, we can provide lower bounds for $\gamma_n^*$. For $1/k \geq p > 1/(k+1)$ with integer $k$, the policy that skips the first $1/e$ fraction of candidates and then makes an offer to any top $k$ candidate afterwards obtains a robust ratio of at least $1/e$. The following result gives improved bounds for $\gamma_\infty^*(p)$.
\begin{theorem}\label{thm:main_lower_bound} Let $p^*\approx 0.594$ be the solution of $p^{{(2-p)}/{(1-p)}} = (1-p)^2$. There is a solution of the infinite LP for $p\geq p^*$ that guarantees $\gamma_n^*\geq \gamma_\infty^*(p) = p^{{p}/{(1-p)}}$. For $p\leq p^*$ we have $\gamma_\infty^*(p) \geq (p^*)^{p^*/(1-p^*)}\approx 0.466$. Moreover, for $p\to 0$, we obtain $\gamma_\infty(p) \geq 0.51$. 
\end{theorem}
To prove this result, we use the following general procedure to construct feasible solutions for the infinite LP.
For any numbers $0 < t_1 \leq t_2 \leq \cdots \leq t_k \leq \cdots \leq 1$, there is a policy that makes offers to any candidate with partial rank $ r_t \in [k] $ when a fraction $t_k$ of the total number of candidates has been observed
(Proposition~\ref{prop:feasible_sol}). For $p\geq p^*$, the policy corresponding to $t_1 = p^{1/(1-p)}$ and $t_2=t_3=\cdots =1$ has a robust ratio of at least $p^{p/(1-p)}$. For $p\leq p^*$, we show how to transform the solution for $p^*$ into a solution for $p$ with an objective value at least as good as the value $\gamma_{\infty}^*(p^*)= (p^*)^{p^*/(1-p^*)}$. For values of $p$ close to $0$, we construct a feasible solution of the infinite LP that guarantees $\gamma_\infty^*(p) \geq 0.51$.

Figure~\ref{fig:current_bounds} depicts the various theoretical bounds we obtain. For reference, we also include numerical results for $\gamma_n^*$ computed by solving $(LP)_{n,p}$ in Theorem~\ref{thm:main_lp} for $n=200$ and with $p$ ranging from $p=10^{-2}$ to $p=1$, with increments of $10^{-3}$. Since $\gamma_n^*$ is nonincreasing in $n$, the numerical values obtained by solving $(LP)_{n,p}$ also provide an upper bound over $\gamma_{\infty}^*$.
\begin{figure}[h!]
	\centering
	\includegraphics[width=0.7\linewidth]{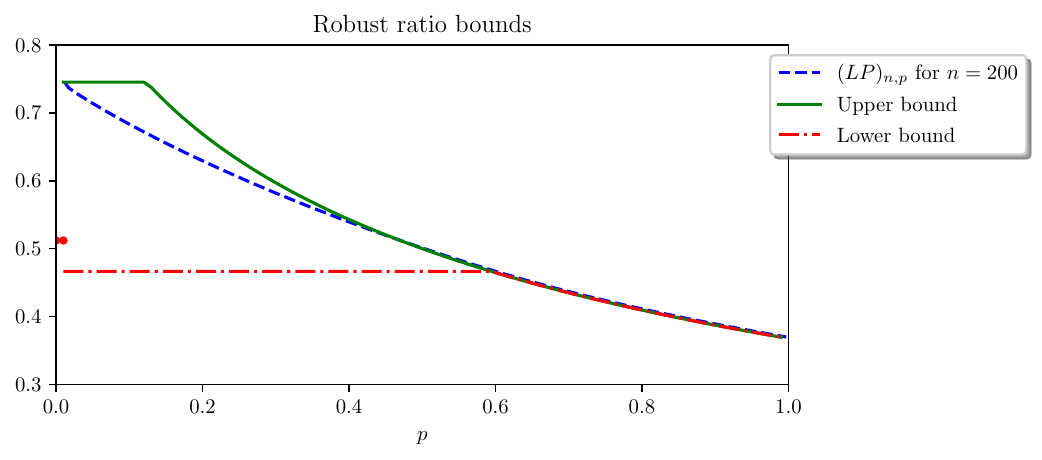}
	\caption{Bounds for $\gamma_\infty^*$ as a function of $p$. The solid line represents the theoretical upper bound given in Theorem~\ref{thm:main_upper_bound}. The dashed-dotted line corresponds to the theoretical lower bound given in Theorem~\ref{thm:main_lower_bound}; for $p$ close to $0$, the guarantee rises to $0.51$. In dashed line we present numerical results by solving $(LP)_{n,p}$ for $n=200$ candidates.}\label{fig:current_bounds}
\end{figure}

We follow this introduction with a brief literature review. In Section~\ref{sec:prelims} we present preliminaries, including MDP notation and an alternative characterization of the robust ratio in terms of utility functions. In Section~\ref{sec:lp_finite_formulation} we present the MDP framework and use it to prove Theorems \ref{thm:main_polyhedron} and \ref{thm:main_lp}. In Section~\ref{sec:infinite_LP} we introduce the infinite relaxation of (LP), then prove Theorem~\ref{thm:main_upper_bound} in Section~\ref{sec:upper_bound}. In Section~\ref{sec:exact_sol_large_p} we prove Theorem~\ref{thm:main_lower_bound}. In Section~\ref{sec:experiments} we present a numerical comparison between the policies obtained by solving $(LP)_{n,p}$ and other benchmarks policies. We conclude in Section \ref{sec:conc}, and an appendix includes proofs and analysis omitted from the main article.

\section{Related Work}\label{sec:lit_review}

\paragraph{Online advertising and online selection} Online advertising has been extensively studied from the viewpoint of two-sided markets: advertisers and platform. There is extensive work in auction mechanisms to select ads (e.g.\ second-price auctions, the VCG mechanism, etc.), and the payment systems between platforms and advertisers (pay-per-click, pay-for-impression, etc.) \citep{devanur2009price,edelman2007internet,fridgeirsdottir2018cost}; see also \citet{choi2020online} for a review. On the other hand, works relating the platform, advertisers, and web users have been studied mainly from a learning perspective, to improve ad targeting \citep{devanur2009price,farahat2012effective,hazan2019introduction}. In this work, we also aim to display an ad to a potentially interested user. Multiple online selection problems have been proposed to display ads in online platforms, e.g., packing models~\citep{babaioff2007knapsack, korula2009algorithms}, secretary problems and auctions~\citep{babaioff2008online}, prophet models \citep{alaei2012online} and online models with ''buyback``~\citep{babaioff2008selling}. In our setting, we add the possibility that a user ignores the ad; see e.g.\ \citet{cho2004people,dreze2003internet}. Failure to click on ads has been considered in full-information models \citep{goyal2019online}; however, our setting considers only partial information, where the rank of an incoming customer can only be assessed relative to previously observed customers---a typical occurrence in many online applications. Our model is also disaggregated and looks at each ad individually. Our goal is to understand the right time to display an ad/make offers via the $\SPUA$ and the robust ratio for each individual ad.

\paragraph{Online algorithms and arrival models} Online algorithms have been extensively studied for adversarial arrivals~\citep{borodin2005online}. This \emph{worst-case} viewpoint gives robust algorithms against any input sequence, which tend to be conservative. 
Conversely, some models assume distributional information about the inputs~\citep{kertz1986stop,kleywegt1998dynamic,lucier2017economic}. The random order model lies in between these two viewpoints, and perhaps the most studied example is the secretary problem \citep{dynkin1963optimum,gilbert1966recognizing,lindley1961dynamic}. Random order models have also been applied in Adword problems \citep{devanur2009adwords}, online LP's \citep{agrawal2014dynamic} and online knapsacks~\citep{babaioff2007knapsack,kesselheim2014primal}, among others.

\paragraph{Secretary problems} Martin Gadner popularized the secretary problem in his 1960 \emph{Mathematical Games} column; for a historical review, see \citet{ferguson1989solved} and also the classical survey by \citet{freeman1983secretary}. For the classical secretary problem, the optimal strategy that observes the first $n/e$ candidates and thereafter selects the best candidate was computed by \citet{lindley1961dynamic,gilbert1966recognizing}. The model has been extensively studied in ordinal/ranked-based settings \citep{lindley1961dynamic,gusein1966problem,vanderbei1980optimal,buchbinder2014secretary} as well as cardinal/value-based settings \citep{bateni2013submodular,kleinberg2005multiple}.

A large body of work has been dedicated to augment the secretary problem. Variations include cardinality constraints \citep{buchbinder2014secretary,vanderbei1980optimal,kleinberg2005multiple}, knapsack constraints \citep{babaioff2007knapsack}, and matroid constraints \citep{soto2013matroid,feldman2014simple,lachish2014log}. Model variants also incorporate different arrival processes, such as Markov chains \citep{hlynka1988secretary} and more general processes \citep{dutting2020secretaries}. Closer to our problem are the data-driven variations of the model \citep{correa2020sample,correa2021secretary,kaplan2020competitive}, where samples from the arriving candidates are provided to the decision maker. Our model can be interpreted as an online version of sampling, where a candidate rejecting the decision maker's offer is tantamount to a sample. This also bears similarity to the exploration-exploitation paradigm often found in online learning and bandit problems \citep{cesa2006prediction,hazan2019introduction,freund1999adaptive}.

\paragraph{Uncertain availability in secretary problems} The $\SPUA$ is studied by \citet{smith1975secretary} with the goal of selecting the top candidate --- $k=1$ in \eqref{eq:robust_ratio_pol} --- who gives an asymptotic probability of success of $p^{1/{(1-p)}}$. If the top candidate rejects the offer, this leads to zero value, which is perhaps excessively pessimistic in scenarios where other competent candidates could accept. \citet{tamaki1991secretary} considers maximizing the probability of selecting the top candidate \emph{among} the candidates that will accept the offer. Although more realistic, this objective still gives zero value when the top candidate that accepts is missed because she arrives early in the sequence. In our approach, we make offers to candidates even if we have already missed the top candidate that accepts the offer; this is also appealing in utility/value-based settings (see Proposition~\ref{prop:utility_characterization}). We also further the understanding of the model and our objective by presenting closed-form solutions and bounds. See also \citet{bruss1987optimal,presman1973best,cowan1979optimal}.

\paragraph{Linear programs in online selection} \modif{Linear programming has been used extensively in online selection \citep{agrawal2014dynamic,beyhaghi2021improved,epstein2024selection,kesselheim2014primal}. Typically, the LP is used as a structured bound over a benchmark that the algorithm designer can compare with. Our approach is different, as we provide an exact formulation of our robust objective. In secretary problems, early work used mostly MDP's \citep{lindley1961dynamic,smith1975secretary,tamaki1991secretary}, while LP formulations were recently introduced by \citet{buchbinder2014secretary}; subsequently, multiple formulations have been used to solve variants of the secretary problem \citep{chan2014revealing,correa2020sample,dutting2020secretaries}. We extend this line of work and use an MDP to derive the exact polyhedron that encodes policies for the $\SPUA$; this helps explain why some LP formulations in secretary problems are exact (see, Subsection~\ref{sec:subsec:warm_up}). \citet{jiang2021tight,jiang2023tightness,perez2022iid} are closer to our work, as these studies characterize the optimal policies for prophet inequalities. Further connections between MDP and LPs in related models have been studied mostly in approximate regimes \citep{adelman2007dynamic,torrico2018polyhedral,torrico2022dynamic} 
and particularly in constrained MDP's \citep{altman1999constrained,haskell2013stochastic,haskell2015convex}. To the best of the authors' knowledge, there was previously no explicit connection between the MDP formulation and the exact LP formulation in secretary problems.}

\section{Preliminaries}\label{sec:prelims}



To discuss our model, we use standard MDP notation for secretary problems \citep{dynkin1963optimum,freeman1983secretary,lindley1961dynamic}. An instance is characterized by the number of candidates $n$ and the probability $p \in (0, 1] $ that an offer is accepted. For $t\in [n]$ and $s\in [t]$, a \emph{state} of the system is a pair $(t,s)$ indicating that the candidate currently being evaluated is the $t$-th and the corresponding partial rank is $r_t=s$. To simplify notation, we add the states $(n+1,s)$, $s\in [n+1]$, and the state $\Theta$ as absorbing states where no decisions can be made. For $t<n$, transitions from a state $(t,s)$ to a state $(t+1,\sigma)$ are determined by the random permutation $\pi=(R_1,\ldots,R_n)$. We denote by $S_t \in \{ (t,s) \}_{s\in [t]}$ the random variable indicating the state in the $t$-th stage. A simple calculation shows
\[
\Prob(S_{t+1}=(t+1,\sigma) \mid S_{t}=(t,s) ) = \Prob(r_{t+1}=\sigma\mid r_t=s) = \Prob(S_{t+1}=(t+1,\sigma) ) =  1/(t + 1),
\]
for $t< n$, $s\in [t]$ and $\sigma\in [t+1]$. In other words, partial ranks at each stage are independent. For notational convenience, we assume the equality also holds for $t=n$. 
Let $\mathbf{A}=\{ \textbf{offer}, \textbf{pass}  \}$ be the set of \emph{actions}. For $t\in [n]$, given a state $(t,s)$ and an action $A_t=a\in \mathbf{A}$, the system transitions to a state $S_{t+1}$ with the following probabilities :
\[
P_{((t,s),a),(\tau,\sigma)}=\Prob(S_{t+1}= (\tau,\sigma)\mid S_t=(t,s), A_t =a ) = \begin{cases}
	\frac{1-p}{t+1} & a=\textbf{offer}, \tau=t+1,\sigma\in [\tau] \\
	p & a=\textbf{offer}, (\tau,\sigma)=\Theta \\
	\frac{1}{t+1} & a= \textbf{pass}, \tau=t+1, \sigma\in [\tau].
\end{cases}
\]
The randomness is over the permutation $\pi$ and the random outcome of the $t$-th candidate's decision. We utilize states $(n+1,\sigma)$ as end states and the state $\Theta$ as the state indicating that an offer is accepted from the state $S_t$. A \emph{policy} $\mathcal{P}: \{ (t,s) : t\in [n], s\in [t]   \} \to \mathbf{A}$ is a function that observes a state $(t,s)$ and decides to extend an offer ($\mathcal{P}(t,s)=\textbf{offer}$) or move to the next candidate ($\mathcal{P}(t,s)=\textbf{pass}$). The policy specifies the actions of a decision maker at any point in time. The \emph{initial state} is $S_1=(1,1)$ and the computation (of a policy) is a sequence of state and actions $(1,1),a_1,(2,s_2),a_2,(3,s_3),\ldots$ where the states transitions according to $P_{((t,s),a),(t+1,\sigma)}$ and $a_t=\mathcal{P}(t,s_t)$. Note that the computation always ends in a state $(n+1,\sigma)$ for some $\sigma$ or the state $\Theta$, either because the policy was able to go through all candidates or because some candidate $t$ accepted an offer.


We say that a policy reaches stage $t$ or reaches the $t$-th stage if the computation of a policy contains a state $s_t=(t,s)$ for some $s\in [t]$. We also refer to stages as \emph{times}.

A randomized policy is a function $\mathcal{P}: \{ (t,s) : t\in [n], s\in [t]   \} \to \Delta_{\mathbf{A}} $ where $\Delta_{\mathbf{A}} = \{ (q,1-q) : q\in [0,1] \}$ is the probability simplex over $\mathbf{A}=\{ \textbf{offer}, \textbf{pass} \}$ and $\mathcal{P}(s_t) = (q_t,1-q_t)$ means that $\mathcal{P}$ selects the \textbf{offer} action with probability $q_t$ and otherwise selects \textbf{pass}.

We could also define policies that remember previously visited states and at state $(t,s_t)$ make decisions based on the \emph{history}, $(1,s_1),\ldots,(t,s_t)$. However, MDP theory guarantees that it suffices to consider Markovian policies, which make decisions based only on $(t,s_t)$; see \citet{puterman2014markov}. 

We say that a policy $\mathcal{P}$ \emph{collects} a candidate with rank $k$ if the policy extends an offer to a candidate that has rank $k$ \emph{and} the candidate accepts the offer. Thus our objective is to find a policy that solves
\begin{align*}
	\gamma_n^* &= \max_{\mathcal{P}} \min_{k\in [n]} \frac{\Prob(\mathcal{P} \text{ collects a candidate with rank}\leq k)}{1-(1-p)^k} \\
	&= \max_{\mathcal{P}} \min_{k\in [n]} \Prob(\mathcal{P} \text{ collects a top $k$ candidate}\mid \text{a top $k$ candidate accepts}).
\end{align*}

%

The following result is an alternative characterization of $\gamma_n^*$ based on utility functions. We use this result to relate $\SPUA$ to the i.i.d.\ prophet inequality problem;
the proof appears in Appendix~\ref{sec:app:prelims_missing}. Consider a nonzero utility function $U:[n]\to \R_+$ with $U_1 \geq U_2 \geq \cdots \geq U_n\geq 0$, and any rank-based algorithm $\ALG$ for the $\SPUA$, i.e., $\ALG$ only makes decisions based on the relative ranking of the values observed. In the value setting, if $\ALG$ collects a candidate with overall rank $i$, it obtains value $U_i$. We denote by $U(\ALG)$ the value collected by such an algorithm.
\begin{proposition}\label{prop:utility_characterization}
	Let $\ALG$ be a $\gamma$-robust algorithm for $\SPUA$. For any $U:[n]\to \R_+$ we have $\E[U(\ALG)]\geq \gamma \E[U(\OPT)]$ where $\OPT$ is the maximum value obtained from candidates that accept. Moreover, \[\gamma_n^* = \max_{\ALG}\min \left\{\frac{\E\left[ U(\ALG)  \right]}{\E\left[ U(\OPT)  \right]} : U:[n]\to \R_+, U \neq 0, U_1 \geq U_2 \geq \cdots \geq U_n \geq 0 \right\} . \]
\end{proposition}


\section{The LP Formulation}\label{sec:lp_finite_formulation}

In this section, we present the proofs of Theorems \ref{thm:main_polyhedron} and \ref{thm:main_lp}. Our framework is based on MDP and can be used to derive similar LPs in the literature, e.g.\ \citet{buchbinder2014secretary,chan2014revealing,correa2020sample}. As a byproduct, we also show that $\gamma_n^*$ is a nonincreasing sequence in $n$ (Lemma~\ref{lem:non_increasing_gamma}). For ease of explanation, we first present the framework for the classical secretary problem, then we sketch the approach for our model. Technical details are deferred to the appendix.

\subsection{Warm-up: MDP to LP in the classical secretary problem}\label{sec:subsec:warm_up}

We next show how to derive an LP for the classical secretary problem \citep{buchbinder2014secretary} using an MDP framework. In this model, the goal is to maximize the probability of choosing the top candidate, and there is no offer uncertainty.
\begin{theorem}[\cite{buchbinder2014secretary}]\label{thm:classical_secretary_problem}
The maximum probability of choosing the top-ranked candidate in the classical secretary problem is given by
\[
\max\left\{ \sum_{t=1}^n \frac{t}{n} x_{t} :  x_{t} \leq \frac{1}{t} \left( 1- \sum_{\tau=1}^{t-1} x_{\tau}  \right), \forall t\in [n], \x\geq 0  \right\}.
\]
\end{theorem}

We show this as follows:
\begin{enumerate}[leftmargin=*]
	\item\label{enum:framework_1} First, we formulate the secretary problem as a Markov decision process, where we aim to find the highest ranked candidate. Let $v^*_{(t,s)}$ be the maximum probability of selecting the highest ranked candidate in $t+1,\ldots,n$ given that the current state is $(t,s)$. We define $v^*_{(n+1,s)}=0$ for any $s$. The value $v^*$ is called the \emph{value function} and it can be computed via the optimality equations~\citep{puterman2014markov}
	\begin{align}
		v^*_{(t,s)} = \max\left\{  \Prob(R_t=1 \mid r_t=s) , \frac{1}{t+1} \sum_{\sigma=1}^t v^*_{(t+1,\sigma)}   \right\} . \label{eq:opt_eq_simple_sp}
	\end{align}
	The first term in the max operator corresponds to the expected value when the \textbf{offer} action is chosen in state $(t,s)$. The second corresponds to the expected value in stage $t+1$ when we decide to \textbf{pass} in $(t, s)$. Note that $\Prob(R_t=1\mid r_t=s)= t/n$ if $s=1$ and $\Prob(R_t=1\mid r_t=s)=0$ otherwise. The optimality equations~\eqref{eq:opt_eq_simple_sp} can be solved via backwards recursion, and $v^*_{(1,1)}\approx 1/e$ (for large $n$). An optimal policy can be obtained from the optimality equations by choosing at each state an action that attains the maximum, breaking ties arbitrarily.
	
	\item\label{enum:framework_2} 
Using a standard argument \citep{manne60}, it follows that $v^*=( v^*_{(t,s)} )_{t,s}$ is an optimal solution of the linear program $(D)$:

\modif{$(D)$ \begin{tabular}{rp{0.87\linewidth}}
	$\displaystyle$&  $\displaystyle \qquad \min_{v\geq 0}  \,\,\, \,\,\,\qquad v_{(1,1)} $ \\
	 &  {\vspace{-1cm}\begin{align}
			v_{(t,s)}  &\geq  \Prob(R_t=1\mid r_t=s) &\forall t\leq n ,\forall s\leq t  \label{eq:const_1} \\
			v_{(t,s)} &\geq   \frac{1}{1+t} \sum_{\sigma=1}^{t+1} v_{(t+1,\sigma)} & \forall t\leq n,s\leq t \label{eq:const_2}
	\end{align}}\vspace{-0.5cm}
\end{tabular}}


	\item\label{enum:framework_3} Taking the dual of $(D)$, we obtain $(P)$:
	
	\modif{$(P)$ \begin{tabular}{rp{0.87\linewidth}}
		$\displaystyle$&  $\qquad \displaystyle\max_{x,y\geq 0} \,\,\,\,\quad\qquad  \sum_{t=1}^n \frac{t}{n}x_{t,1}$ \\
		&  {\vspace{-1cm}\begin{align}
				x_{1,1} + y_{1,1}  &\leq  1 \label{eq:P_const_1} \\
				x_{t,s} + y_{t,s} & \leq \frac{1}{t} \sum_{\sigma=1}^{t-1} y_{t-1,\sigma} & \forall t \leq n,s\leq t \label{eq:P_const_2}
		\end{align}}\vspace{-0.5cm}
	\end{tabular}}
	
	Variables $x_{t,s}$ are associated with constraints~\eqref{eq:const_1}, $y_{t,s}$ with constraints~\eqref{eq:const_2}. Take any solution $(\x,\y)$ of the problem $(P)$ and note that the objective does not depend on $\y$. Incrementing $\y$ to tighten all constraints does not alter the feasibility of the solution and the objective does not change; thus we can assume that all constraints are tight in $(P)$. Here, $x_{t,s}$ is the probability that a policy (determined by $\x,\y$) reaches the state $(t,s)$ and makes an offer, while $y_{t,s}$ is the probability that the same policy reaches state $(t,s)$ but decides to not to make an offer.
	
	\item\label{enum:framework_4} Finally, projecting the feasible region of $(P)$ onto the variables $(x_{t,1})_{t\in [n]}$, e.g.\ via Fourier-Motzkin elimination (see \citet{schrijver1998theory} for a definition), gives us Theorem~\ref{thm:classical_secretary_problem}. We skip this for brevity.
\end{enumerate}
The same framework can be applied to obtain the linear program for the secretary problem with \emph{rehire}~\citep{buchbinder2014secretary} and the formulation for the \emph{$(J,K)$-secretary problem} \citep{buchbinder2014secretary,chan2014revealing}. It can also be used to derive an alternative proof of the result by~\cite{smith1975secretary}. Besides secretary problems, this approach using MDP has been applied for prophet problems \citep{jiang2021tight} and in online bipartite matching~\citep{torrico2018polyhedral,torrico2022dynamic}. 

\subsection{Framework for the $\SPUA$}

Next, we sketch the proof of Theorem~\ref{thm:main_polyhedron} and use it to derive Theorem~\ref{thm:main_lp}. Technical details are deferred to the appendix.

In the classical secretary problem, the objective is to maximize the probability of choosing the top candidate, which we can write in the recursion of the value function $v^*$. For our model, the objective $\gamma_{\mathcal{P}}$ corresponds to a multi-objective criteria, and it is not clear a priori how to write the objective as a reward. We present a two-step approach: (1) First, we follow the previous subsection's argument to uncover the polyhedron of policies; (2) second, we show that our objective function can be written in terms of variables in this polyhedron, and we maximize this objective in the polyhedron.

\subsubsection{The polyhedron of policies via a generic utility function}

When we obtained the dual LP $(P)$ (Step~\ref{enum:framework_3} in the above framework), anything related to the objective of the MDP is moved to the objective value of the LP, while anything related to the actions of the MDP remained in constraints \eqref{eq:P_const_1}-\eqref{eq:P_const_2}. This suggests using a generic utility function to uncover the space of policies. Consider any vector $U:[n]\to \R_+$, and suppose that our objective is to maximize the utility collected, where choosing a candidate of rank $i$ means obtaining $U_i\geq 0$ value. Let $v^*_{(t,s)}$ be the maximum value collected in times $t,t+1,\ldots,n$ given that the current state is $(t,s)$, where $v^*_{(n+1,s)}=0$. Then, the optimality equations yield
\begin{align}
	v^*_{(t,s)} = \max \left\{  p U_t(s) + (1-p) \frac{1}{t+1} \sum_{\sigma=1}^{t+1} v^*_{(t+1,\sigma)}, \frac{1}{t+1} \sum_{\sigma=1}^{t+1} v^*_{(t+1,\sigma)}   \right\} , \label{eq:mdp_opt_0}
\end{align}
where $U_t(s)= \sum_{i=1}^n U_i \Prob(R_t=i \mid r_t=s)$. The term in the left side of the max operator is the expected value obtained by an \textbf{offer} action, while the term in the right corresponds to the expected value of the \textbf{pass} action. Using an approach similar to the one used in steps~\ref{enum:framework_2} and \ref{enum:framework_3} from the previous subsection, we can deduce that
\[
\textsc{Pol} = \left\{ (\x,\y)\geq 0 : x_{1,1}+ y_{1,1} =1, x_{t,s}+ y_{t,s} = \frac{1}{t}\sum_{\sigma=1}^{t-1}\left(  y_{t-1,\sigma} + (1-p)x_{t-1,\sigma}  \right) , \forall t > 1, s\in [t] \right\}
\]
contains all policies (Theorem~\ref{thm:main_polyhedron}). A formal proof is presented in the appendix.

\subsubsection{The linear program}

Next, we consider Theorem~\ref{thm:main_lp}. Given a policy $\mathcal{P}$, we define $x_{t,s}$ to be the probability of reaching state $(t,s)$ and making an offer to the candidate, and $y_{t,s}$ to be the probability of reaching $(t,s)$ and passing. Then $(\x,\y)$ belongs to $\textsc{Pol}$. Moreover, 
\begin{align}
	\Prob(\mathcal{P} \text{ collects a top } k \text{ candidate}) = p \sum_{t=1}^n \sum_{s=1}^t x_{t,s}\Prob(R_t\leq k \mid r_t=s). \label{eq:pol_collect_equals_lineear_obj}
\end{align}
Conversely, any point $(\x,\y)\in \textsc{Pol}$ defines a policy $\mathcal{P}$: At state $(t,s)$, it extends an \textbf{offer} to the $t$-th candidate with probability $x_{1,1}$ if $t=1$, or probability ${t x_{t,s}}/{\left(\sum_{\sigma=1}^{t-1} y_{t-1,\sigma} + (1-p)x_{t-1,\sigma} \right)}$ if $t>1$. Also, $\mathcal{P}$ satisfies \eqref{eq:pol_collect_equals_lineear_obj}. Thus,
\begin{align}
	\gamma_n^* & = \max_{\mathcal{P}} \min_{k\in [n]} \frac{\Prob(\mathcal{P} \text{ collects a top } k \text{ candidate})}{1-(1-p)^k} \nonumber\\
	& = \max_{(x,y)\in \textsc{Pol}} \min_{k\in [n]} \frac{ p \sum_{t=1}^n \sum_{s=1}^t x_{t,s}\Prob(R_t\leq k \mid r_t=s)}{1-(1-p)^k} \nonumber\\
	& = \max \left\{\gamma:  (\x,\y)\in \textsc{Pol}, \gamma \leq \frac{ p \sum_{t=1}^n \sum_{s=1}^t x_{t,s}\Prob(R_t\leq k \mid r_t=s)}{1-(1-p)^k}, \forall k\in [n]  \right\}. \label{eq:robust_ratio_pol_form}
\end{align}
Projecting the feasible region of \eqref{eq:robust_ratio_pol_form} as in step~\ref{enum:framework_4} onto the $(\x,\gamma)$-variables gives us Theorem~\ref{thm:main_lp}. The details appear in Appendix~\ref{sec:app:lp_via_mdp}.

Our MDP framework also allows us to show the following monotonicity result.
\begin{lemma}\label{lem:non_increasing_gamma}
	For a fixed $p\in (0,1]$, we have $\gamma_n^*\geq \gamma_{n+1}^*$ for any $n\geq 1$.
\end{lemma}
We sketch the proof of this result here and defer the details to Appendix~\ref{subsec:app:gamma_decreasing_in_n}. The dual of the LP~\eqref{eq:robust_ratio_pol_form} can be reformulated as

\resizebox{\columnwidth}{!}{ \hfil \begin{tabular}{rrp{0.7\linewidth}}
		&$\displaystyle\min_{\substack{u : [n]\to \R_+ \\ \sum_i u_i \geq 1}}$ & \quad \quad $v_{(1,1)}$ \\
		(DLP)&s.t. &  {\vspace{-1.5cm}\begin{align*}
				v_{(t,s)} & \geq \max\left\{  U_t(s) + \frac{1-p}{t+1} \sum_{\sigma=1}^{t=1} v_{(t+1,\sigma)} , \frac{1}{t+1} \sum_{\sigma=1}^{t=1} v_{(t+1,\sigma)} \right\} & \forall t\in [n], s\in [t] \\
				v_{(n+1,s)}& =0 & \forall s\in [n+1],
			\end{align*}\vspace{-0.8cm}}
\end{tabular}}
where $U_t(s)=  p\sum_{j=1}^n \left( \sum_{k\geq j} u_k /{\left(1-(1-p)^k\right)} \right)\Prob(R_t=j\mid r_t=s)$. The variables $u_1,\ldots,u_n$ correspond to the constraints involving $\gamma$ in the LP~\eqref{eq:robust_ratio_pol_form}. Note that (DLP) is the minimum value that an MDP can attain when the utility functions are given by $U_i = p \sum_{k\geq i} u_k /{\left(1-(1-p)^k\right)}$. Taking any weighting $u:[n]\to \R_+$ with $\sum_{i} u_i \geq 1$, we extend it to $\widehat{u}:[n+1]\to \R_+$ by setting $\widehat{u}_{n+1}=0$. We define accordingly $\widehat{U}_i = p \sum_{k\geq i} \widehat{u}_k /{\left(1-(1-p)^k\right)}$, and note that $U_i=\widehat{U}_i$ for $i\leq n$ and $\widehat{U}_{n+1}=0$. Using a coupling argument, from any policy for utilities $\widehat{U}$ with $n+1$ candidates, we can construct a policy for utilities $U$, with $n$ candidates, where both policies collect the same utility. Thus, the utility collected by the optimal policy for $U$ upper bounds the utility collected by an optimal policy for $\widehat{U}$. The conclusion follows since $\gamma_{n+1}^*$ is a lower bound for the latter value.
%

Since $\gamma_n^*\in [0,1]$ and $(\gamma_n^*)_n$ is a monotone sequence in $n$, $ \lim_{n \to \infty} \gamma_n^*$ must exist. In the next section we show that the limit corresponds to the value of a continuous LP.

\section{The Continuous LP}\label{sec:infinite_LP}

In this section we introduce the continuous linear program $(CLP)$, and we show that its value $\gamma_\infty^*$ corresponds to the limit of $\gamma_n^*$ when $n$ tends to infinity. We also state Proposition~\ref{prop:feasible_sol}, which allows us to construct feasible solutions of $(CLP)$ using any set of times $0 < t_1 \leq t_2 \leq \cdots \leq 1$. In the finite model, the solution constructed in this section has the natural interpretation of segmenting time: for a candidate arriving between times $t_i n$ and $t_{i+1} n$, we make an offer if the candidate has partial rank $i$ or better. 
In the remainder of the section, \emph{finite model} refers to the $\SPUA$ with $n <\infty$ candidates, while the \emph{infinite model} refers to $\SPUA$ when $n \to \infty$.

We assume $p\in(0,1]$ fixed. The continuous LP $(CLP)$ is an infinite linear program with variables given by a function $\alpha:[0,1]\times \N \to [0,1]$ and a scalar $\gamma\geq 0$. Intuitively, if in the finite model we interpret $x_{t,s}$ as weights and the sums of $x_{t,s}$ over $t$ as Riemann sums, then the limit of the finite model, should have a robust ratio computed by the continuous LP $(CLP)$

{ $(CLP)_{p}$ \begin{tabular}{cp{0.75\linewidth}}
		$\displaystyle\sup_{\substack{\alpha:[0,1]\times \N \to [0,1] \\ \gamma \geq 0}}$ &  $\quad \quad \gamma$ \\
		s.t. &  {\vspace{-1.5cm} \begin{align}
				t\alpha (t, s) & \leq 1 - p \int_0^t \sum_{\sigma \geq 1} \alpha(\tau, \sigma) \,  \mathrm{d}\tau  & \forall t\in [0,1], s\geq 1 \label{const:infinity_dynamic} \\
				\gamma&\leq  \frac{p \int_0^1  \sum_{s\geq 1} \alpha(t,s) \sum_{\ell=s}^k \binom{\ell-1}{s-1} t^s (1-t)^{\ell-s} \, \mathrm{d}t  }{(1-(1-p)^k)} & \forall k\geq 1 . \label{const:infinity_multiobj}
			\end{align}\vspace{-0.5cm}}
\end{tabular}}

We denote by $\gamma_\infty^*=\gamma_\infty^*(p)$ the objective value of $(CLP)_p$. The following result formalizes the fact that the value of the continuous LP $(CLP)_p$ is in fact the robust ratio of the infinite model. The proof is similar to other continuous approximations~\citep{chan2014revealing}; a small caveat in the proof is the restriction of the finite LP to the top $(\log n )/ p$ candidates, as they carry most of the weight in the objective function. The proof is deferred to Appendix~\ref{sec:app:infinite_lp_approximation}.

\begin{lemma}\label{lem:approximation_inf_LP_finite_LP}
	Let $\gamma_n^*$ be the optimal robust ratio for $n$ candidates and let $\gamma_{\infty}^*$ be the value of the continuous LP $(CLP)_p$. Then $|\gamma_n^* - \gamma_\infty^*| \leq \mathcal{O}\left( {(\log n)^2}/{\left(p\sqrt{n}\right)}  \right)$.
\end{lemma}

The following proposition gives a recipe to find feasible solutions for $(CLP)_p$. We use it to construct lower bounds in the following sections.

\begin{proposition}\label{prop:feasible_sol}
	Consider $0\leq t_1 \leq t_2 \leq \cdots \leq 1$ and consider the function $\alpha: [0,1]\times \N \to [0,1]$ defined such that for $t\in [t_i, t_{i+1})$
	\[
	\alpha(t,s) = \begin{cases}
		{T_i}/{t^{i\cdot p + 1}} & s\leq i \\
		0 & s > i,
	\end{cases}
	\]
	where $T_i= (t_1\cdots t_i)^p$. Then $\alpha$ satisfies Constraint~\eqref{const:infinity_dynamic}.
\end{proposition}

\begin{proof}
	We verify that inequality~\eqref{const:infinity_dynamic} holds. We only need to verify it for $t\in[t_i,t_{i+1})$ with $i\geq 1$ since $\alpha(t,s)=0$ for $t\in [0,t_1)$. We define $t_0=0$ and $T_0=0$. For $t\in [t_i,t_{i+1})$ we have
	\begin{align*}
		1- p \int_0^t \sum_{\sigma\geq 1} \alpha(\tau, \sigma) \,\mathrm{d} \tau & = 1 - p \left(\sum_{j= 1}^{j-1} \int_{t_j}^{t_{j+1}} j \frac{T_j}{\tau^{jp+1}} \, \mathrm{d}\tau\right) - p \int_{t_i}^t i \frac{T_i}{\tau^{ip + 1}} \, \mathrm{d}\tau \\
		& = 1- p \sum_{j=1}^{i-1} j\cdot T_j \cdot \frac{1}{-p j}\left( t_{j+1}^{-jp} - t_j^{-jp}  \right) - p i \cdot T_i\cdot \frac{1}{-ip}\left( t^{-ip} - t_{i}^{ip} \right)\\
		& = 1 + \sum_{j=1}^{i-1} T_j \left( t_{j+1}^{-jp} - t_{j}^{-jp} \right) + T_i \left( t^{-ip} - t_i^{ip} \right) \\
		& = 1 + \left(\sum_{j=1}^{i-1} T_{j+1} t_{j+1}^{-(j+1)p} - T_j t_j^{-jp}\right) + T_i \left(  t^{-ip}  - t_i^{-ip}\right)\tag{Since $T_{j} t_{j+1}^{-jp} = T_{j+1} t_{j+1}^{-(j+1)p}$} \\
		& = 1 + T_{i} t_i^{-ip} - T_1 t_1^{-p} + T_i\left( t^{-ip} - t_i^{-ip} \right) \\
		& = T_i t^{-ip} \geq t \alpha(t,s)
	\end{align*}
	for any $s\geq 1$. This concludes the proof.
\end{proof}

We use this result to show lower bounds for $\gamma_\infty^*$. For instance, if $1/k\geq p > 1/(k+1)$ for some integer $k$, and we set $t_1= 1/e$ and $t_2=t_3=\cdots = 1$, we can show that $\gamma_\infty^*(p)$ is at least $1/e$. Thus, in combination with Lemma~\ref{lem:non_increasing_gamma}, we have that $\gamma_n^*(p) \geq 1/e$ for any $n$ and $p>0$; we skip this analysis for brevity. In Section~\ref{sec:exact_sol_large_p}, we use Proposition~\ref{prop:feasible_sol} to show exact solutions of $\gamma_\infty^*$ for large $p$.

%

\section{Upper Bounds for the Continuous LP}\label{sec:upper_bound}

We now consider upper bounds for $(CLP)$ and prove Theorem~\ref{thm:main_upper_bound}, which states that $\gamma_\infty^*(p)\leq \min \left\{ p^{p/(1-p)} , 1/\beta \right\}$, for any $p\in (0,1]$, where $1/\beta\approx 0.745$ and $\beta$ is the unique solution of $\int_0^1 {(y(1-\log y)+ \beta-1)^{-1}} {\mathrm{d}y} = 1$ \citep{kertz1986stop}.

We show that $\gamma_\infty^*$ is bounded by each term in the minimum operator. For the first bound, we have
\begin{align*}
	\gamma_n^* & = \max_{\mathcal{P}} \min_{k\in [n]} \frac{\Prob(\mathcal{P} \text{ collects a top $k$ candidate})}{1-(1-p)^k}\\
	& \leq \max_{\mathcal{P}} \frac{\Prob(\mathcal{P} \text{ collects the top candidate})}{p} .
\end{align*}
The probability of collecting the highest candidate in $\SPUA$ is shown by \citet{smith1975secretary} to be $p^{1/(1-p)} + o(1)$, where $o(1)\to 0$ as $n\to \infty$. Thus, by Lemma~\ref{lem:non_increasing_gamma}, we have
\[
\gamma_\infty^* (p)\leq \gamma_n^* (p) \leq p^{{p}/{(1-p)}} + o(1)/p.
\]
Taking the limit $n \to \infty$, we conclude $\gamma_\infty^* (p)\leq p^{p/(1-p)}$.

For the second bound, we use the following technical result; its proof is deferred to Appendix~\ref{sec:app:upper_bounds}, but we give a short explanation here. A $\gamma$-robust algorithm $\mathcal{A}$ for the $\SPUA$, in expectation, has $pn$ candidates to choose from and $(1-p)n$ candidates from which the algorithm can learn about candidate quality. We give an algorithm $\mathcal{A}'$ that solves the i.i.d.\ prophet inequality for any $m\approx pn$ i.i.d.\ random variables $X_1,\ldots,X_m$, for $m$ large. The algorithm $\mathcal{A}'$ runs a utility version of $\mathcal{A}$ in $n$ values sampled from the distribution $X_1$ (see the discussion before Proposition~\ref{prop:utility_characterization}), guaranteeing at least a factor $\gamma$ of the maximum of $m\approx pn$ of these samples, which is the value of the prophet. $\mathcal{A}'$ is the capped utility version of $\mathcal{A}$, where no more than $m\approx pn$ offers can be made. Using concentration bounds, we show that the loss of these restrictions is minimal. \citet{kaplan2020competitive} uses a similar argument, with the difference that their sampling is a fixed fraction of the input and is done in advance, while in our case the sampling is online and might deviate from the expectation, implying the need for concentration bounds. The following result summarizes the reduction and the upper bound.
\begin{theorem}\label{lem:reduction_sec_to_proph}
	Let $p\in (0,1)$ and $\mathcal{A}$ be any algorithm that is $\gamma$-robust for the $\SPUA$ for any $n$. Then $\gamma \leq 1/\beta$, where $\beta\approx 1.34$ is the unique solution of the integral equation $\int_0^1 ( y(1-\log y)+ (\beta-1) )^{-1} \, \mathrm{d}y=1$.
%
%
\end{theorem}

The proof of Theorem~\ref{lem:reduction_sec_to_proph} uses the reduction mentioned in the previous paragraph. The use of concentration bounds guarantees a $(1-o(1))\gamma$ multiplicative approximation with $e^{-\Theta(n^2)}$ additive error for the i.i.d.\ prophet inequality problem. Specifically, for any distribution $\mathcal{D}$ over $[0,1]$, we can guarantee
\begin{align}
	\E\left[\mathrm{Val}(\mathcal{A}')\right] + e^{-\Theta(n^2)} \geq \gamma (1- o(1))\E\left[\max_{i\leq m} X_i\right], \label{ineq:prophet_red}
\end{align}
where $X_1,\ldots,X_m$ distribute according to $\mathcal{D}$ and $\mathrm{Val}(\mathcal{A}')$ is the value collected by $\mathcal{A}'$, the algorithm described in the previous paragraph where no more than $m\approx pn$ offers are made. Here $o(1)$ is a term that can be chosen arbitrarily close to $0$ for $n$ large enough. Unfortunately, we cannot conclude that $\gamma\leq 1/\beta$ immediately from this bound, since this bound only holds for multiplicative error in the i.i.d.\ prophet problem. We bypass this technical challenge as follows. A combination of results by~\cite{hill1982comparisons} and~\cite{kertz1986stop} shows that, for any $m$ and for $\varepsilon'>0$ small enough, there is an i.i.d.\ instance $X_1,\ldots,X_m$ with support in $[0,1]$ such that
\[
\E\left[ \max_{i\leq m} X_i \right] \geq (a_m- \varepsilon') \sup\left\{ \E[X_\tau] : \tau \in T_m   \right\},
\]
where $T_m$ is the class of stopping times for $X_1,\ldots,X_m$, and $a_m\to \beta$. Thus, using Inequality~\ref{ineq:prophet_red},  we must have
\[
\gamma (1-o(1))  \leq \frac{1}{a_m - \varepsilon'}  + \frac{e^{-\Theta(n^2)}}{\E\left[ \max_{i\leq m} X_i  \right]}
\]
for $m  \approx pn$. A slight reformulation of \citeauthor{hill1982comparisons}'s result allows us to set $\varepsilon'= 1/m^3$ and $\E[\max_{i\leq m} X_i] \geq 1/m^3$ (see the discussion at the end of Appendix~\ref{sec:app:upper_bounds}). Thus, as $n\to \infty$ we have $m\to \infty$ and so $e^{-\Theta(n^2)}/\E[\max_{i\leq m}X_i] \to 0$. In the limit we obtain $\gamma(1-o(1)) \leq 1/\beta$, which implies our stated result.

An algorithm that solves $(LP)_{n,p}$ and implements the policy given by the solution is $\gamma_\infty^*$-robust (Theorem~\ref{thm:main_lp} and the fact that $\gamma_n^* \geq \gamma_\infty^*$) for any $n$. Thus, by the previous analysis we obtain $\gamma_\infty^*\leq 1/\beta\approx 0.745$.

%
%
%
%
%
%
%

\section{Lower Bounds for the Continuous LP}\label{sec:exact_sol_large_p}

In this section we consider lower bounds for $(CLP)_p$ and prove Theorem~\ref{thm:main_lower_bound}. We first give optimal solutions of $(CLP)_p$ for large values of $p$. For $p\geq p^* \approx 0.594$, the optimal value of $(CLP)_p$ is $\gamma_\infty^*(p)= p^{p/(1-p)}$ and the optimal strategy is to observe $p^{1/(1-p)}$ fraction of the candidates, and then make offers to the best observed candidate so far. We then show that for $p\leq p^*$, $\gamma_{\infty}^*(p)\geq (p^*)^{p^*/(1-p^*)}\approx 0.466$. At the end of the section, we show that $\gamma_\infty^*(p) \geq 0.51$ when $p\to 0$.

\subsection{Exact solution for large $p$}

We now show that for $p\geq p^*$, $\gamma_\infty^*(p)= p^{p/(1-p)}$, where $p^*\approx 0.594$ is the solution of $(1-p)^2 = p^{(2-p)/(1-p)}$. Thanks to the upper bound $\gamma_\infty^*(p)\leq p^{p/(1-p)}$ for any $p\in (0,1]$, it is enough to exhibit feasible solutions $(\alpha,\gamma)$ of the continuous LP $(CLP)_p$ with $\gamma \geq p^{p/(1-p)}$.

Let $t_1=p^{1/(1-p)}$, $t_2=t_3=\cdots=1$, and consider the function $\alpha$ defined by $t_1,t_2,\ldots$ in Proposition~\ref{prop:feasible_sol}. That is, for $t\in [0,p^{1/(1-p)})$, $\alpha(t,s)=0$ for any $s\geq 1$ and for $t\in [p^{1/(1-p)},1]$ we have
\[
\alpha(t,s) = \begin{cases}
	{p^{{p}/{(1-p)}}}/{t^{1+p}} & s=1 \\
	0 & s > 1 .
\end{cases}
\]
Let $\gamma \doteq \inf_{k\geq 1} {p}{\left(1-(1-p)^k\right)^{-1}} \int_{p^{{1}/{(1-p)}}}^1 \frac{p^{{p}/{(1-p)}}}{t^{1+p}} \sum_{\ell=1}^k t(1-t)^{\ell-1}  \mathrm{d}t$. Then $(\alpha,\gamma)$ is feasible for the continuous LP $(CLP)_p$, and we aim to show that $\gamma\geq p^{p/(1-p)}$ when $p\geq p^*$. The result follows by the following lemma.
\begin{lemma}\label{lem:ratio_bound}
	For any $p\geq p^*$ and any $\ell\geq 0$, $\int_{p^{{1}/{(1-p)}}}^1 {(1-t)^\ell} t^{-p} \, \mathrm{d}t \geq (1-p)^\ell$.
\end{lemma}

We defer the proof of this lemma to Appendix~\ref{subsec:app:exact_sol_large_p}. Now, we have
\begin{align*}
	\gamma & = \inf_{k\geq 1}\frac{p}{1-(1-p)^k} \int_{p^{{1}/{1-p}}}^1 \frac{p^{{p}/{(1-p)}}}{t^{1+p}} \sum_{\ell=1}^k t(1-t)^{\ell-1} \mathrm{d}t \\
	& = p^{{p}/{(1-p)}} \inf_{k\geq 1} \frac{\sum_{\ell=1}^{k} \int_{p^{1/(1-p)}}^1 {t^{-p}} (1-t)^{\ell-1} \, \mathrm{d}t }{\sum_{\ell=1}^k (1-p)^{\ell-1}} \\
	& \geq p^{{p}/{(1-p)}} \inf_{k\geq 1} \inf_{\ell\in[k]} \int_0^1 \frac{1}{t^p} \frac{(1-t)^{\ell-1}}{(1-p)^{\ell-1}} \, \mathrm{d}t \geq p^{{p}/{(1-p)}} ,
\end{align*}
where we use the known inequality ${\sum_{\ell=1}^m a_\ell}/{\sum_{\ell=1}^m b_\ell}\geq \min_{\ell\in [m]} {a_\ell}/{b_\ell}$ for $a_\ell,b_\ell >0$, for any $\ell$, and the lemma. This shows that $\gamma_\infty^* \geq p^{p/(1-p)}$ for $p\geq p^*$.

\begin{remark}
	Our analysis is tight. For $k=2$, constraint
	$$\frac{p}{1-(1-p)^2} \int_{p^{{1}/{(1-p)}}}^1 \frac{p^{{p}/{(1-p)}}}{t^{1+p}} \sum_{\ell=1}^k t(1-t)^{\ell-1}\, \mathrm{d}t \geq p^{{p}/{(1-p)}}$$
	holds if and only if $p\geq p^*$.
\end{remark}

\subsection{Lower bounds for small $p$}

In this subsection we present two lower bounds for $\gamma_\infty^*(p)$ when $p\leq p^*$, with $p^*\approx 0.594$ as obtained in the previous subsection. The first bound guarantees $\gamma_\infty^*(p) \geq (p^*)^{p^*/(1-p^*)}\approx 0.466$; the second guarantees $\gamma_\infty^*(p) \geq 0.51$ when $p$ approaches $0$. We present details for the first bound, as it includes a mechanism to transform the solution of $(CLP)_{p^*}$ into a solution of $(CLP)_p$ for $p\leq p^*$. We defer some details of the latter bound, as it uses a construction similar to~\citep{correa2021secretary} with a different limit argument. 

Let $\varepsilon \in [0,1)$ satisfy $p=(1-\varepsilon)p^*$. For the argument, we take the solution $\alpha^*$ for $(CLP)_{p^*}$ that we obtained in the last subsection and we construct a feasible solution for $(CLP)_{p}$ with objective value at least $(p^*)^{p^*/(1-p^*)}$. For simplicity, we denote $\tau^*= (p^*)^{1/(1-p^*)}$.

From the previous subsection, we know that the optimal solution $\alpha^*$ of $(CLP)_{p^*}$ has the following form. For $t\in [0,\tau^*)$, $\alpha^*(t,s)=0$ for any $s$, while for $t\in [\tau^*,1]$ we have
\[
\alpha^*(t,s) = \begin{cases}
	{(p^*)^{p^*/(1-p^*)}}/{t^{p^*+1}} & s=1 \\
	0 & s> 1 .
\end{cases}
\]
For $(CLP)_{p}$, we construct a solution $\alpha$ as follows. Let $\alpha(t,s)= \varepsilon^{s-1}\alpha^*(t,1)$ for any $t\in [0,1]$ and $s\geq 1$; for example, $\alpha(t,1)= \alpha^*(t,1)$. If we interpret $\alpha^*$ as a policy, it only makes offers to the highest candidate observed. By contrast, in $(CLP)_p$ the policy implied by $\alpha$ makes offers to more candidates (after time $\tau^*$), with a probability geometrically decreasing according to the relative ranking of the candidate. 
\begin{lemma}
The solution $\alpha$ satisfies constraints~\eqref{const:infinity_dynamic}, \[t \alpha(t,s) \leq  1- p \int_0^t \sum_{\sigma\geq 1} \alpha(\tau,\sigma)\mathrm{d}\tau , \] for any $t\in[0,1]$, $s\geq 1$.
\end{lemma}

\begin{proof}
	Indeed,
	\begin{align*}
		1- p \int_0^t \sum_{\sigma\geq 1} \alpha(\tau,\sigma) \mathrm{d}\tau & = 1- p^*(1-\varepsilon)\int_0^t \sum_{\sigma \geq 1} \varepsilon^{\sigma-1} \alpha^*(\tau,1) \mathrm{d}\tau \\
		& = 1 - p^*\int_0^t \alpha^*(\tau, 1) \mathrm{d}\tau \tag{Since $\sum_{\sigma\geq 1} \varepsilon^{\sigma-1} = 1/(1-\varepsilon)$} \\
		& = 1 - p^* \int_0^t \sum_{\sigma\geq 1} \alpha^*(\tau,\sigma) \, \mathrm{d}\tau \tag{Since $\alpha^*(\tau,\sigma)=0$ for $\sigma>1$}\\
		& \geq t\alpha^*(t,1) . \tag{By feasibility of $\alpha^*$}
	\end{align*}
	Since $\alpha(t,s) = \varepsilon^{s-1}\alpha^*(t,1) \leq \alpha^*(t,1)$, we conclude that $\alpha$ satisfies \eqref{const:infinity_dynamic} for any $t$ and $s$.
\end{proof}

We now define $\gamma=\inf_{k\geq 1} {p}{\left(1-(1-p)^k\right)^{-1}} \int_0^1 \sum_{s\geq 1} \alpha(t,s) \sum_{\ell=s}^k \binom{\ell-1}{s-1} t^s (1-t)^{\ell-s}\mathrm{d}t$. Using the claim, we know that $(\alpha,\gamma)$ is feasible for $(CLP)_p$, and need to verify that $\gamma \geq (p^*)^{p^*/(1-p^*)}$. Similar to the analysis in the previous section, the result follows by the following result.
\begin{lemma}\label{claim:ratio_bound_small_p}
	For any $\ell \geq 0$, $\int_{\tau^*}^1 {(1 - (1-\varepsilon)t)^\ell }{t^{-p^*}} \, \mathrm{d}t \geq (1-p)^\ell$.
\end{lemma}

Before proving the claim, we establish the bound:
\begin{align*}
	\gamma &  = \inf_{k\geq 1} \frac{1}{\sum_{\ell=1}^k (1-p)^{\ell-1}} \int_0^1 \sum_{s=1}^k \varepsilon^{s-1} \alpha^*(s, 1) \sum_{\ell=s}^k \binom{\ell-1}{s-1} t^s(1-t)^{\ell-s} \, \mathrm{d}t \tag{Using definition of $\alpha$} \\
	& = (p^*)^{{p^*}/{(1-p^*)}} \inf_{k\geq 1} \frac{1}{\sum_{\ell=1}^k (1-p)^{\ell-1}} \int_{\tau^*}^1 \frac{1}{t^{p^*+1}}  \sum_{\ell=1}^k \sum_{s=1}^\ell \varepsilon^{s-1} \binom{\ell-1}{s-1} t^s(1-t)^{\ell-s} \, \mathrm{d}t \tag{Using the definition of $\alpha^*$ and changing order of summmation} \\
	& = (p^*)^{{p^*}/{(1-p^*)}} \inf_{k\geq 1} \frac{1}{\sum_{\ell=1}^k (1-p)^{\ell-1}} \int_{\tau^*}^1 \frac{1}{t^{p^*}}  \sum_{\ell=1}^k (1- (1-\varepsilon) t)^{\ell-1} \, \mathrm{d}t \tag{Using the binomial expansion} \\
	& = (p^*)^{{p^*}/{(1-p^*)}} \inf_{k\geq 1} \frac{\sum_{\ell=1}^k \int_{\tau^*}^1 {t^{-p^*}}  (1- (1-\varepsilon) t)^{\ell - 1} \, \mathrm{d}t}{\sum_{\ell=1}^k (1-p)^{\ell-1}}   \geq (p^*)^{{p^*}/{(1-p^*)}}
\end{align*}
We again used the inequality ${\sum_{\ell=1}^m a_\ell}/{\sum_{\ell=1}^m b_\ell}\geq \min_{\ell\in [m]} {a_\ell}/{b_\ell}$ for $a_\ell,b_\ell >0$, for any $\ell$, and the claim.

\begin{proof}[Proof of Lemma~\ref{claim:ratio_bound_small_p}]
	We have $1-(1-\varepsilon)t = (1-\varepsilon)(1-t) + \varepsilon$. Therefore,
	\begin{align*}
		\int_{\tau^*}^1 \frac{1}{t^{p^*}}(1 - (1-\varepsilon)t)^\ell \mathrm{d}t & = \int_{\tau^*}^1 \frac{1}{t^{p^*}} \sum_{j=0}^{\ell} \binom{\ell}{j} (1-\varepsilon)^{\ell-j}(1-t)^{\ell-j} \varepsilon^{j} \mathrm{d}t \tag{Binomial expansion} \\
		& = \sum_{j=0}^\ell \binom{\ell}{j} (1-\varepsilon)^{\ell-j} \varepsilon^j \int_{\tau^*}^1 \frac{1}{t^{p^*}} (1-t)^{\ell-j} \, \mathrm{d}t \\
		& \geq \sum_{j=0}^\ell \binom{\ell}{j} (1-\varepsilon)^{\ell-j} \varepsilon^j (1-p^*)^{\ell-j} \, \mathrm{d}t \tag{Using Lemma~\ref{lem:ratio_bound} for $p^*$} \\
		& = (\varepsilon + (1-\varepsilon)(1-p^*))^\ell \tag{Using binomial expansion} \\
		& = (1- (1-\varepsilon)p^*)^\ell = (1-p)^\ell  ,
	\end{align*}
	where we used  $p=(1-\varepsilon)p^*$. From this inequality the claim follows.
\end{proof}

\subsubsection{Improved bound for $p$ close to $0$}

Now we present a better bound for $\gamma_\infty^*(p)$, for $p$ close to $0$, using an explicit construction of a solution $\alpha$ for $(CLP)_p$. 

The following proposition gives a sufficient condition to ensure lower bounds over $\gamma_\infty^*(p)$.

\begin{proposition}\label{prop:necessary_lower_bound}
	Let $0\leq t_1\leq t_2 \leq t_3\leq \cdots \leq 1$. If for all $k\geq 1$ we have
	\[
	T_k \left( \frac{t_{k+1}^{1-kp} - t_k^{1-kp}}{1-kp}  \right) \geq \gamma p(1-p)^{k-1} ,
	\]
	where $T_k=(t_1\cdots t_k)^p$, then, $\gamma_\infty^*(p) \geq \gamma$.
\end{proposition}

\begin{proof}
	Let $\alpha(t,s)$ be defined as follows. For $i\geq 1$, if $t\in [t_i,t_{i+1})$, then
	\[
	\alpha(t,s) = \begin{cases}
		T_i / t^{i\cdot p + 1} & s\leq i\\
		0 & s> i.
	\end{cases}
	\]
	Then, by Proposition~\ref{prop:feasible_sol}, we know that $\alpha$ defines a feasible solution to $(CLP)_p$. Now, using $\alpha(t,s)\geq \alpha(t,\ell)$ for any $s\leq \ell$, we obtain
	\begin{align*}
		p\int_0^1 \sum_{s\geq 1} \alpha(t,s) \sum_{\ell=s}^k \binom{\ell-1}{s-1} t^s (1-t)^{\ell-s} \, \mathrm{d}t & \geq p \int_0^1 \sum_{\ell=1}^k \alpha(t,\ell) t \, \mathrm{d}t \\
		& = p \sum_{i=1}^\infty \min\{ k,i \} T_i \left( \frac{t_{i+1}^{1-ip} - t_i^{1-ip} }{1-ip}  \right) \\
		& = p \sum_{j=1}^k \sum_{i=i}^\infty  T_i \left( \frac{t_{i+1}^{1-ip} - t_i^{1-ip} }{1-ip}  \right) \\
		& \geq p \sum_{j=1}^k \sum_{i=j}^\infty \gamma p (1-p)^{j-1} \\
		& =  \gamma p \sum_{j=1}^k (1-p)^{j-1} = \gamma (1-(1-p)^k) .
	\end{align*}
	This holds for any $k$; thus, $\gamma_\infty^*(p) \geq \gamma$.
\end{proof}

We now present an iterative method to generate a sequence $t_1,t_2,\ldots$ as in Proposition~\ref{prop:necessary_lower_bound}. Fix $t_1\in [0,1]$ and $\gamma >0$, and define $A_k=t_1 (1-p)^{-k} + \gamma p k$ for $k\geq 1$. Note that $A_k$ is increasing in $k$. Define $t_2 = t_1 \left( 1+ \gamma p(1-p)/t_1 \right)^{1/(1-p)}$, and for $k\geq 2$, define $t_{k+1}$ as follows:
\begin{align}
	\left( \frac{t_{k+1}}{t_k}  \right)^{1-kp} = \frac{A_k(1-p)}{A_{k-1}}. \label{eq:definition_of_tk}
\end{align}

%

\begin{lemma}
	The sequence $t_1,\ldots$ defined above satisfies $t_k\leq t_{k+1}$ for each $k\geq 1$. Moreover, for any $k\geq 1$,
	\[
	T_k \left( \frac{t_{k+1}^{1-kp} - t_k^{1-kp}}{1-kp}  \right) = \gamma p(1-p)^{k-1} ,
	\]
	where $T_k=(t_1\cdots t_k)^p$.
\end{lemma}

\begin{proof}
	The first part follows from the fact that $k < 1/p$ if and only if $A_k (1-p) \geq A_{k-1}$. For the second part, let
	\[
	B_k = T_k \left( \frac{t_{k+1}^{1-kp} - t_k^{1-kp}}{1-kp}  \right).
	\]
	It is easy to verify that $B_1=\gamma p$ using the definition of $t_2$. Now, for $k\geq 2$,
	\begin{align*}
		\frac{B_{k+1}}{B_{k}} & = t_{k+1}^p \left(  \frac{t_{k+2}^{1-(k+1)p}-t_{k+1}^{1-(k+1)p}}{t_{k+1}^{1-kp} - t_{k}^{1-kp}}  \right) \left( \frac{1-kp}{1-(k+1)p} \right).
	\end{align*}
	Using Identity~\eqref{eq:definition_of_tk} and $A_j(1-p) - A_{j-1} = \gamma p (1-jp)$, we obtain
	\begin{align*}
		\frac{B_{k+1}}{B_{k}} & = t_{k+1}^p \frac{A_{k-1}}{A_k} \frac{t_{k+1}^{1-(k+1)p}}{t_k^{1-kp}} = (1-p).
	\end{align*}
	Then, inductively, we can show that $B_k = \gamma p (1-p)^{k-1}$.
\end{proof}

Our goal is to give $t_1\in [0,1]$ and $\gamma\in [0,1]$  as large as possible such that $\lim_k t_k \leq 1$, with $t_k$ defined via~\eqref{eq:definition_of_tk}. 


\begin{lemma}\label{lem:limit_log_tk}
	We have
	\[
	\lim_{k\to \infty} \log t_{k} = \log t_1 + \int_0^\infty \frac{\gamma}{t_1 e^x + \gamma x} \, \mathrm{d}x.
	\]
	for $p\to 0$.
\end{lemma}

The proof of the lemma is technical and borrows a strategy used by~\cite{correa2021secretary}. The main insight is to apply logarithms to both sides of Identity~\eqref{eq:definition_of_tk}, and find an expression for $\log t_{k+1}$ as a sum of terms only involving $A_j$, $\gamma$ and $p$. In the limit in $k$, and then in $p$, we can reinterpret the sums as Riemann sums that later translate into the integral term given in the lemma. We defer the details of this proof to the Appendix.

We find numerically that the combination of $t_1\approx 0.228$ and  $\gamma\approx0.511$ ensures $\lim_{k\to \infty} \log t_k = 0$. Thus, $\gamma_\infty^*(p)\geq 0.51$ for $p$ close to $0$; note that $t_1$ remains bounded away from $0$. This means that the policy given by the values $t_1,t_2,\ldots$ spends the first $t_1$ fraction of time ``exploring'' before ``exploiting'', and this exploring time remains a constant. 

\section{Computational Experiments}\label{sec:experiments}

In this section we aim to empirically test our policy; to do so, we focus on utility models. Recall from  Proposition~\ref{prop:utility_characterization} that a $\gamma$-robust policy ensures at least $\gamma$ fraction of the optimal offline utility, for any utility function that is consistent with the ranking, i.e., $U_j < U_i$ if and only if $i \prec j$. This is advantageous for practical scenarios, where a candidate's ``value'' may be unknown to the decision maker.

We evaluate the performance of two groups of solutions. The first group includes policies that are computed without the knowledge of any utility function:
\begin{itemize}[leftmargin=*]
	\item Robust policy $(\robpol(n,p))$, corresponds to the optimal policy obtained by solving $(LP)_{n,p}$.
	
	\item \citeauthor{tamaki1991secretary}'s policy $(\tamapol(n,p))$, which maximizes the probability of selecting the best candidate willing to accept an offer. To be precise, \citet{tamaki1991secretary} studies two models of availability: MODEL 1, where the availability of the candidate is known after an offer has been made; and MODEL 2, where the availability of the candidate is known upon the candidate's arrival. MODEL 2 has higher values and is computationally less expensive to compute; we use this policy. Note that in $\SPUA$, the expected value obtained by learning the availability of the candidate after making an offer is the same value obtained in the model that learns the availability upon arrival. Therefore, MODEL 2 is a better model to compare our solutions to than MODEL 1.
\end{itemize}
In the other group, we have policies that are computed with knowledge of the utility function.
\begin{itemize}[leftmargin=*]
	\item The expected optimal offline value $(\E[U(\OPT(U,n,p))])$, which knows the outcome of the offers and the utility function. It can be computed via $\sum_{i=1}^n U_i p (1-p)^{i-1}$. For simplicity, we write $\OPT$ when the parameters are clear from the context.
	
	\item The optimal rank-based policy if the utility function is known in advance, $(\utilpol(U,n,p))$, computed by solving the optimality equation
	\[
	v_{(t,s)} = \max\left\{ U_t(s) + \frac{1-p}{t+1} \sum_{\sigma=1}^{t+1} v_{(t+1,\sigma)}, \frac{1}{t+1} \sum_{\sigma=1}^{t+1} v_{(t+1,\sigma)}, \right\} ,
	\]
	with boundary condition $v_{(n+1,\sigma)}=0$ for any $\sigma$. We write $\utilpol(n,p)$ when $U$ is clear from the context. We use a rank-based policy as opposed to a value-based policy for computational efficiency. 
\end{itemize}
Note that $\E [U(\robpol)], \E [U(\tamapol)]\leq \E [U(\utilpol)]\leq \E[U(\OPT)]$ and by Proposition~\ref{prop:utility_characterization}, $E[U(\robpol)] \geq \gamma_{n}^* \E[U(\mathcal{A})]$ for any $\mathcal{A}$ of the aforementioned policies.

We consider the following decreasing utility functions:
\begin{itemize}[leftmargin=*]
	\item \textbf{Top $k$ candidates are valuable (top-$k$)}. For $k\in [n]$, we consider utility functions of the form $U_i= 1 + \varepsilon^{i}$ for $i\in [k]$ and $U_i=\varepsilon^{i}$ for $i>k$ with $\varepsilon = 1/n$. Intuitively, we aim to capture the notion of an elite set of candidates, where candidates outside the top $k$ are not nearly as appealing to the decision maker. For instance, renowned brands like to target certain members of a population for their ads. 
We test $k= 1,2,3,4$.
	
	\item \textbf{Power law population}. $U_i = i^{-1/(1+\delta)}$ for $i\geq 1$ and small $\delta>0$. Empirical studies have shown that the distribution of individual performances in many areas follows a power law or Pareto distribution \citep{clauset2009power}. If we select a random person from $[n]$, the probability that this individual has a performance score of at least $ t$ is proportional to $ t^{-(1+\delta)}$. We test $\delta \in \{ 10^{-2}, 10^{-1}, 2\cdot 10^{-1} \}$.
\end{itemize}
We run experiments for $n=200$ candidates and range the probability of acceptance $p$ from $p=10^{-2}$ to $p=9\cdot 10^{-1}$.

\subsection{Results for top-$k$ utility function}

In this subsection, we present the results for utility function that has largest values in the top $k$ candidates, where $k=1,2,3,4$. In  Figure~\ref{fig:graphu10200}, we plot the ratio between the value collected by $\mathcal{A}$ and $\E[U(\OPT)]$, for $\mathcal{A}$ being $\utilpol, \robpol$ and $\tamapol$.

\begin{figure}[h!]
	\centering
	\includegraphics[width=0.49\linewidth]{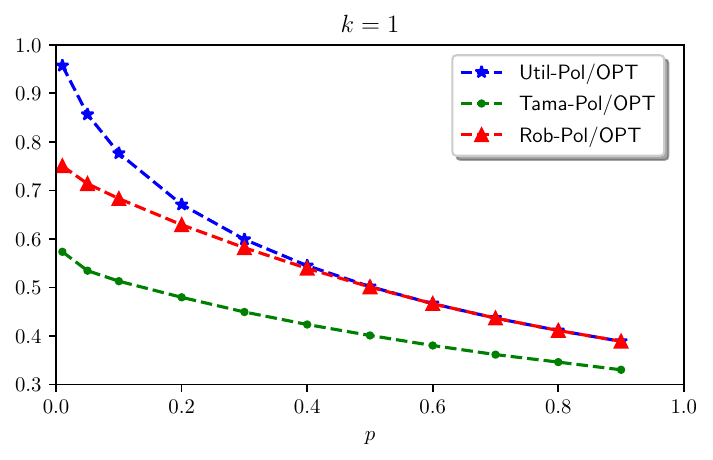}
	\includegraphics[width=0.49\linewidth]{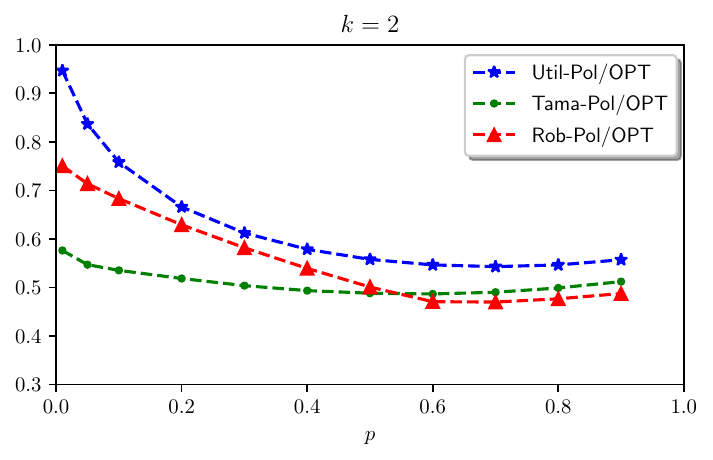}
	\includegraphics[width=0.49\linewidth]{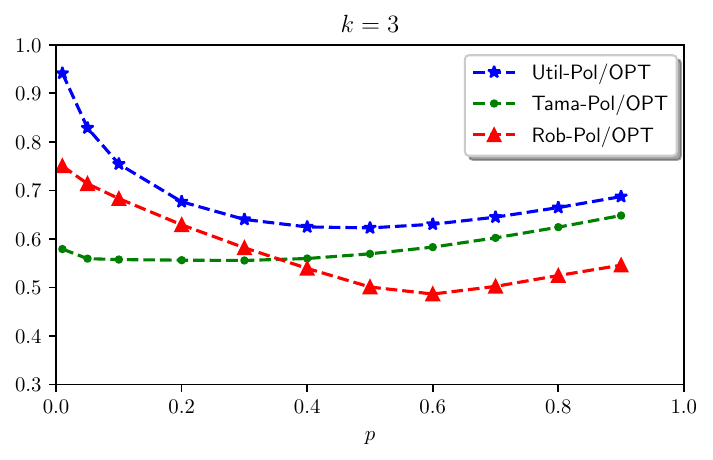}
	\includegraphics[width=0.49\linewidth]{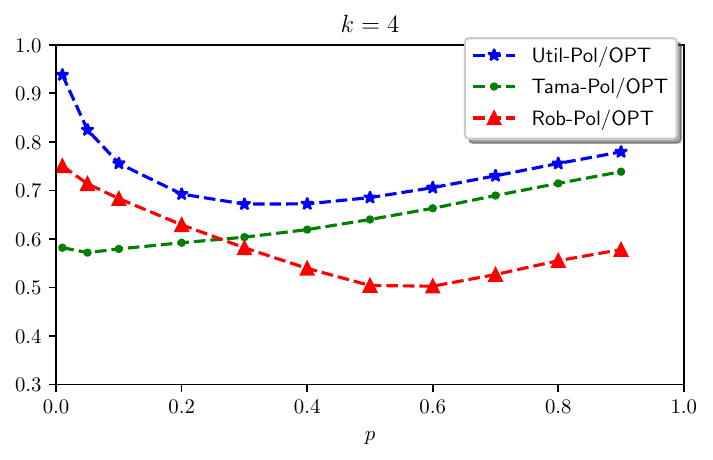}
	\caption{Approximation factors for the top $k$ utility function, for $k=1,2,3,4$.}
	\label{fig:graphu10200}
\end{figure}

Naturally, of all sequential policies, $\utilpol$ attains the largest approximation factor of $\E[U(\OPT)]$. We observe empirically that $\robpol$ collects larger values than $\tamapol$ for smaller values of $k$. Interestingly, we observe in the four experiments that the approximation factor for $\robpol$ is always better than $\tamapol$ for small values of $p$. In other words, robustness helps online selection problems when the probability of acceptance is relatively low. In general, for this utility function, we observe in the experiments that $\robpol$ collects at least $50\%$ of the optimal offline value, except for the case $k=1$. As $n$ increases (not shown in the figures), we observe that the approximation factors of all three policies decrease; this is consistent with the fact that $\gamma_n^*$, the optimal robust ratio, is decreasing in $n$.

\subsection{Results for power law utility function}

In this subsection, we present the result of our experiments for the power law utility function $U_i = i^{-(1+\delta)}$ for $\delta=10^{-2}, 10^{-1}$ and $2\cdot 10^{-1}$. In Figure~\ref{fig:graphu22200}, we display the approximation factors of the three sequential policies.

\begin{figure}[h!]
	\centering
	\includegraphics[width=0.49\linewidth]{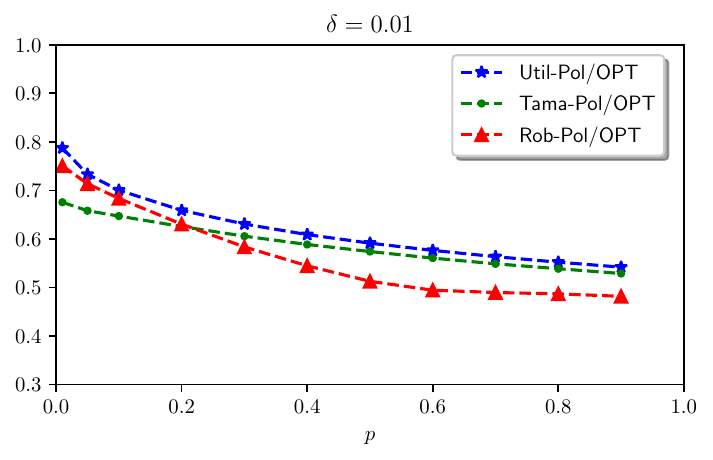}
	\includegraphics[width=0.49\linewidth]{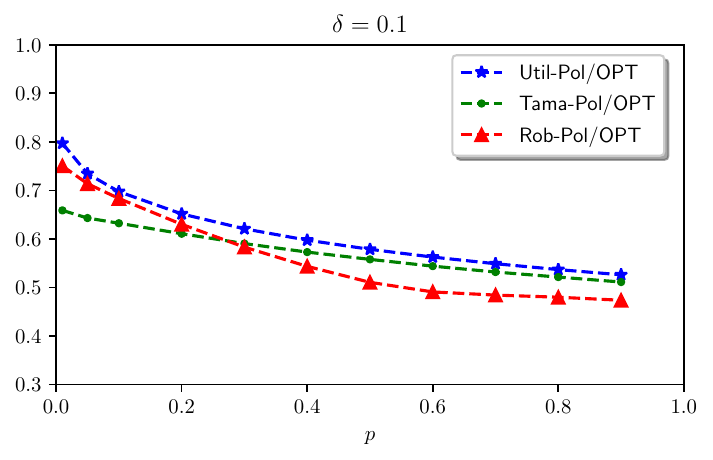}
	\includegraphics[width=0.49\linewidth]{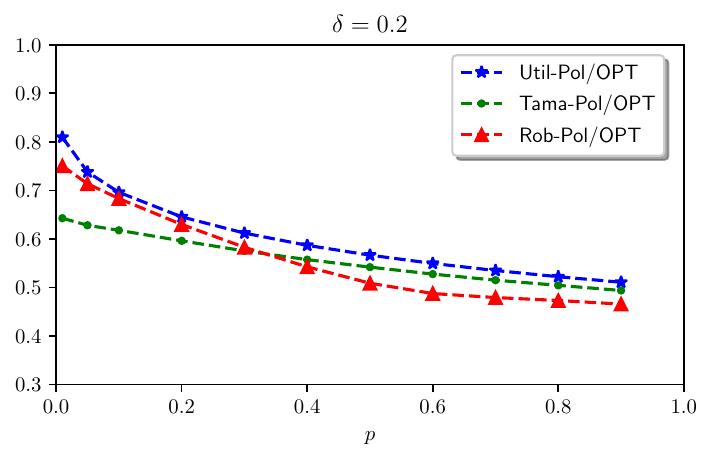}
	\caption{Approximation factor for the power law utility function. The function has the form $U_i=i^{-(1+\delta)}$. Experiments are run for $\delta \in \{ 10^{-2}, 10^{-1}, 2\cdot 10^{-1} \}$.}
	\label{fig:graphu22200}
\end{figure}

Again, we note that $\utilpol$ collects the largest fraction of all sequential policies. We also observe a similar behavior as in the case of the top-$k$ utility function. For small values of $p$, $\robpol$ empirically collects more value than $\tamapol$. As $p$ increases, the largest valued candidate is more willing to accept an offer; hence, \citeauthor{tamaki1991secretary}'s \citeyear{tamaki1991secretary} policy is able to capture that candidate.

In general, our experiments suggests that $\robpol$ is better than $\tamapol$ for smaller values of $p$. This may be of interest in applications where the probability of acceptance is small, say $20\%$ or less.
For instance, some sources state that click-through rates (the fraction of time that an ad is clicked on) are typically less than $1\%$ \citep{farahat2012effective}. Therefore, ad display policies based on $\robpol$ may be more appropriate than other alternatives.

\section{Concluding Remarks}
\label{sec:conc}

We have studied the $\SPUA$, which models an online selection problem where candidates can reject an offer. We introduced the robust ratio as a metric that tries to simultaneously maximize the probability of successfully selecting one of the best $k$ candidates given that at least one of these will accept an offer, for all values of $k$. This objective captures the worst-case scenario for an online policy against an offline adversary that knows in advance which candidates will accept an offer.
We also demonstrated a connection between this robust ratio and online selection with utility functions. We presented a framework based on MDP theory to derive a linear program that computes the optimal robust ratio and its optimal policy. This framework can be generalized and used in other secretary problems (Section~\ref{sec:subsec:warm_up}), for instance, by augmenting the state space. Furthermore, using the MDP framework, we were able to show that the robust ratio $\gamma_n^*$ is a decreasing function in $n$. This enabled us to make connections between early works in secretary problems and recent advances. To study our LP, we allow the number of candidates to go to infinity and obtain a continuous LP. We provide bounds for this continuous LP, and optimal solutions for large $p$.

We empirically observe that the robust ratio $\gamma_n^*(p)$ is convex and decreasing as a function of $p$, and thus we expect the same behavior from $\gamma_{\infty}(p)$, though this remains to be proved (see Figure~\ref{fig:current_bounds}). Based on numerical values obtained by solving $(LP)_{n,p}$, we conjecture that $\lim_{p\to 0} \gamma_\infty^*(p) = 1/\beta \approx 0.745$. This limit is also observed in a similar model \citep{correa2020sample}, where a fraction of the input is given in advance to the decision maker as a sample. In our model, if we interpret the rejection from a candidate as a sample, then in the limit both models might behave similarly. Numerical comparisons between our policies and benchmarks suggest that our proposed policies perform especially well in situations where the probability of acceptance is small, say less than $20\%$, as in the case of online advertisement.

A natural extension is the value-based model, where candidates reveal numerical values instead of partial rankings. Our algorithms are rank-based and guarantee an expected value at least a fraction $\gamma_n^*(p)$ of the optimal offline expected value (Proposition~\ref{prop:utility_characterization}). Nonetheless, algorithms based on numerical values may attain higher expected values than the ones guaranteed by our algorithm. In fact, a threshold algorithm based on sampling may perhaps be enough to guarantee better values, although this requires an instance-dependent approach. The policies we consider are instance-agnostic, can be computed once and used for any input sequence of values. In this value-based model, we would like to consider other arrivals processes. A popular arrival model is the adversarial arrival, where an adversary constructs values and the arrival order in response to the DM's algorithm. Unfortunately, a construction similar to the one in \citet{marchetti1995stochastic} for the online knapsack problem shows that it is impossible to attain a finite competitive ratio in an adversarial regime.

Customers belonging to different demographic groups may have different willingness to click on ads~\citep{cheng2010personalized}. In this work, we considered a uniform probability of acceptance, and our techniques do not apply directly in the case of different probabilities. In ad display, one way to cope with different probabilities depending on customers' demographic group is the following. Upon observing a customer, a random variable (independent of the ranking of the candidate) signals the group of the customer. The probability of acceptance of a candidate depends on the candidate's group. Assuming independence between the rankings and the demographic group allows us to learn nothing about the global quality of the candidates beyond what we can learn from the partial rank. Using the framework presented in this work, with an augmented state space  (time, partial rank, group type), we can write an LP that solves this problem exactly. Nevertheless, understanding the robust ratio in this new setting and providing a closed-form policy are still open questions.

Another interesting extension is the case of multiple selections. In practice, platforms can display the same ad to more than one user, and some job posts require more than one person for a position. In this setting, the robust ratio is less informative. If $k$ is the number of possible selections, one possible objective is to maximize the number of top $k$ candidates selected. We can apply the framework from this work to obtain an optimal LP. Although there is an optimal solution, no simple closed-form strategies have been found even for $p=1$; see e.g.\ \citet{buchbinder2014secretary}).




{\bibliographystyle{plainnat}\fontsize{10}{0}\selectfont\bibliography{biblio}
}

\small
\appendix

\section{Appendix}\label{sec:Appendix}

\subsection{Missing proofs from Section~\ref{sec:prelims}}\label{sec:app:prelims_missing}

\begin{proof}[Proof of Proposition~\ref{prop:utility_characterization}]
	Let $\ALG$ be a $\gamma$-robust algorithm. Fix any algorithm $\ALG$ and any $U_1 \geq \cdots \geq U_n\geq 0$. Let $\varepsilon>0$ and let $\widehat{U}_i= U_i + \varepsilon^{i}$. Thus $\widehat{U}_i > \widehat{U}_{i+1}$ and so rank and utility are in one-to-one correspondence. Then
	\begin{align*}
		\frac{\E[\widehat{U}(\ALG)]}{\E[\widehat{U}(\OPT)]} & \geq \min_{x \in \{ \widehat{U}_1,\ldots, \widehat{U}_n  \}} \frac{\Prob(\widehat{U}(\ALG) \geq x)}{\Prob(\widehat{U}(\OPT) \geq x)}  = \min_{k\leq n} \frac{\Prob(\ALG \text{ collects a top $k$ candidate})}{\Prob(\text{A top $k$ candidate accepts})}  \geq \gamma
	\end{align*}
	where we used the fact that $\ALG$ is $\gamma$-robust. Notice that $\E[U(OPT)]\leq \E[\widehat{U}(OPT)]$, and also $\E[\widehat{U}(\ALG)]\leq \E[U(\ALG)]+ \varepsilon$. Thus doing $\varepsilon \to 0$ we obtain
	\begin{align}
		\frac{\E[U(\ALG)]}{\E[U(\OPT)]} \geq \gamma \label{ineq:utility_gamma_robust}
	\end{align}
	for any nonzero vector $U$ with $U_1\geq U_2 \geq \cdots \geq U_n\geq 0$. This finishes the first part. For the second part, let
	\[
	\overline{\gamma}_n = \min_{\ALG}\max \left\{ \frac{\E[U(\ALG)]}{\E[U(\OPT)]} : U:[n]\to \R_+, U_1\geq \cdots\geq U_n \right\}.
	\]
	Note that the RHS of Inequality~\eqref{ineq:utility_gamma_robust} is independent of $U$, thus minimizing in $U$ in the LHS and then maximizing in $\ALG$ on both sides we obtain $\overline{\gamma}_n \geq \gamma_n^*$. 
	
	To show the reverse inequality, fix $k\in [n]$ and let $\widehat{U}:[n]\to \R_+$ given by $\widehat{U}_i=1$ for $i\leq k$ and $\widehat{U}_i=0$ for $i > k$. Then,
	\[
	\frac{\Prob(\ALG \text{ collects a top $k$ candidate})}{\Prob(\text{A top $k$ candidate accepts})} = \frac{\E[\widehat{U}(\ALG)]}{\E[\widehat{U}(\OPT)]} \geq \min_{\substack{U:[n]\to \R_+\\ U_1 \geq \cdots \geq U_n}} \frac{\E[U(\ALG)]}{\E[U(\OPT)]}.
	\]
	This bound holds for any $k$, thus minimizing over $k$ and then maximizing over $\ALG$ on both sides, we obtain $\gamma_n^*\geq \overline{\gamma}_n$, which finishes the second part.
\end{proof}

\subsection{Missing proofs from Section~\ref{sec:lp_finite_formulation}}\label{sec:app:lp_via_mdp}

Here we present a detailed derivation of Theorem~\ref{thm:main_polyhedron} and Theorem~\ref{thm:main_lp} by revisiting Section~\ref{sec:lp_finite_formulation}.


As stated in Section~\ref{sec:lp_finite_formulation}, we are going to proceed in two stages: (1) First, using a generic utility function, we uncover the space of policies $\textsc{Pol}$. (2) Second, we show that our objective is a concave linear function of the variables of the space of policies that allows us to optimize it over $\textsc{Pol}$.

\subsubsection{Stage 1: The space of policies}

Let $U:[n]\to \R_+$ be an arbitrary utility function and suppose that a DM makes decisions based on partial rankings and her goal is to maximize the utility obtained, where she gets $U_i$ if she is able to collect a candidate of ranking $i$. Let $v^*_{(t,s)}$ be the optimal value obtained by a DM in $\{ t,t+1,\ldots,n \}$ if she is currently observing the $t$-th candidate and this candidate has partial ranking $r_t=s$, i.e., the current state of the system is $s_t=(t,s)$. The function $v^*_{(t,s)}$ is called the \emph{value function}. We define $v^*_{(n+1,\sigma)}$ for any $\sigma\in [n+1]$. Then, by optimality equations~\citep{puterman2014markov}, we must have
\begin{align}
	v^*_{(t,s)} = \max \left\{  p U_t(s) + (1-p) \frac{1}{t+1} \sum_{\sigma=1}^{t+1} v^*_{(t+1,\sigma)}, \frac{1}{t+1} \sum_{\sigma=1}^{t+1} v^*_{(t+1,\sigma)}   \right\} \label{eq:mdp_opt_1}
\end{align}
with $v_{(n+1,s)}^*=0$ for any $s\in [n+1]$. The first part in the max operator corresponds to the expected value obtained by selecting the current candidate, while the second part in the operator corresponds to the expected value collected by passing to the next candidate. Here $U_t(s)=\sum_{i=1}^n U_i \Prob(R_t=i\mid r_t=s)$ is the expected value collected by the DM given that the current candidate has partial ranking $r_t=s$ and accepts the offer. Although an optimal policy for this arbitrary utility function can be computed via the optimality equations, we are more interested in all the possible policies that can be obtained via these formulations. For this, we are going to use linear programming. This has been used in MDP theory~\citep{puterman2014markov,altman1999constrained} to study the space of policies.

\begin{figure}[h!]
	\centering
	\resizebox{\columnwidth}{!}{ \begin{tabular}{lp{3in}|lp{3in}}
			$(D)$ & $  \displaystyle \min_{v\geq 0} \quad v_{(1,1)} $ & $(P)$ & $\displaystyle\max_{\x,\y\geq 0}\quad  \sum_{t=1}^n \sum_{s=1}^t U_t(s) x_{t,s}$ \\
			s.t. & \vspace{-1cm}{\begin{align}
					v_{(t,s)} & \geq p U_t(s) + \frac{1-p}{t+1} \sum_{\sigma=1}^{t+1}  v_{(t+1,\sigma)} \nonumber\\ & \forall t\in [n],s\in [t] \label{const:mdp_lp_1}\\
					v_{(t,s)} & \geq \frac{1}{t+1} \sum_{\sigma=1}^{t+1} v_{(t+1,\sigma)} \nonumber \\
					& \forall t\in [n],s\in [t] \label{const:mpd_lp_2}
				\end{align}\vspace{-0.5cm}} & s.t. & \vspace{-0.8cm} {\begin{align}
					x_{1,1} + y_{1,1} & \leq 1 \label{const:mdp_dlp_1}\\
					x_{t,s} + y_{t,s} & \leq \frac{1}{t}\left( \sum_{\sigma=1}^{t-1} y_{t-1,\sigma} + (1-p) x_{t-1,\sigma} \right)  \nonumber \\
					& \forall t\in [n],s\in [t] \label{const:mpd_dlp_2}
				\end{align}\vspace{-0.5cm}}
	\end{tabular}}
	\caption{Linear program that finds value function $v^*$ for $\SPUA$ and its dual.}\label{fig:LP_SPUA}
\end{figure}

The following proposition shows that the solution of the optimality equations~\eqref{eq:mdp_opt_1} solves the LP $(D)$ in Figure~\ref{fig:LP_SPUA}. We denote by $v_{(D)}$ the value of the LP $(D)$.

\begin{proposition}
	Let $v^*=(v^*_{(t,s)})_{t,s}$ be a solution of~\eqref{eq:mdp_opt_1}, then $v^*$ is an optimal solution of the problem of $(D)$ in Figure~\ref{fig:LP_SPUA}.
\end{proposition}

\begin{proof}
	Since $v^*$ satisfies the optimality equation~\eqref{eq:mdp_opt_1} then it clearly satisfies constraints~\eqref{const:mdp_lp_1} and~\eqref{const:mpd_lp_2}. Thus, $v^*$ is feasible and so $v^*_{(1,1)}\geq v_{(D)}$.
	
	To show the optimality of $v^*$, we show that any solution $\overline{v}$ of the LP is an upper bound for the value function: $v^*\leq \overline{v}$. To show this, we proceed by backward induction in $t=n+1,n,\ldots,1$ and we prove that $v^*_{(t,s)}\leq \overline{v}_{(t,s)}$ for any $s\in [t]$.
	
	We start with the case $t=n+1$. In this case $v^*_{(n+1,s)}=0$ for any $s$ and since $\overline{v}{(n+1,s)}\geq 0$ for any $s$, then the result follows. 
	
	Suppose the result is true for $t=\tau+1,\ldots,n+1$ and let us show it for $t=\tau$. Using Constraints~\eqref{const:mdp_lp_1}-\eqref{const:mpd_lp_2} we must have
	\begin{align*}
		\overline{v}_{(\tau,s)} & \geq \max \left\{  p U_\tau(s) + (1-p) \frac{1}{\tau +1} \sum_{\sigma=1}^{\tau+1} \overline{v}_{(\tau+1,\sigma)}, \frac{1}{\tau+1} \sum_{\sigma=1}^{\tau+1} \overline{v}_{(\tau+1,\sigma)}   \right\} \\
		& \geq \max \left\{  p U_\tau(s) + (1-p) \frac{1}{\tau+1} \sum_{\sigma=1}^{\tau+1} v^*_{(\tau+1,\sigma)}, \frac{1}{\tau+1} \sum_{\sigma=1}^{\tau+1} v^*_{(\tau+1,\sigma)}   \right\} \tag{backward induction}\\
		& = v^*_{(\tau,s)}
	\end{align*}
	where the last line follows by the optimality equations~\eqref{eq:mdp_opt_1}. Thus, $v_{(D)}= \overline{v}_{(1,1)} \geq v^*_{(1,1)}$.
\end{proof}

The dual of the LP (D) is depicted in Figure~\ref{fig:LP_SPUA} and named (P). The crucial fact to notice here is that the feasible region of (P) is oblivious of the utility function (or rewards) given initially to the MDP. This suggest that the region
\[
\textsc{Pol} = \left\{ (\x,\y)\geq 0 : x_{1,1}+ y_{1,1} =1, x_{t,s}+ y_{t,s} = \frac{1}{t}\sum_{\sigma=1}^{t-1}\left(  y_{t-1,\sigma} + (1-p)x_{t-1,\sigma}  \right) , \forall t\in [n], s\in [t] \right\}
\]
codifies all possible policies. The following two propositions formalize this. 

%

\begin{proposition}\label{prop:policies_are_points}
	For any policy $\mathcal{P}$ for the $\SPUA$, consider $x_{t,s}=\Prob(\mathcal{P}\text{ reaches state }(t,s) \text{, selects candidate})$ and $y_{t,s}= \Prob(\mathcal{P}\text{ reaches state }(t,s) \text{, does not select candidate})$. Then $(\x,\y)$ belongs to $\textsc{Pol}$.
\end{proposition}

\begin{proof}
	Consider the event $D_t=\{ t\text{-th candidate turns down offer}  \}$. Then $p=\Prob(D_t)$. Consider also the events
	$$O_t=\{ \mathcal{P} \text{ reaches }t\text{-th candidate and extends an offer} \}$$
	and
	$$\overline{O}_t = \{ \mathcal{P} \text{ reaches }t\text{-th candidate and does not extend offer} \}.$$
	Then $O_t$ and $\overline{O}_t$ are disjoint events and $O_t\cup \overline{O}_t$ equals the event of $\mathcal{P}$ reaching stage $t$. Thus
	\begin{align}
		\mathbf{1}_{O_t \cap \{ S_t=(t,s)  \}} + \mathbf{1}_{\overline{O}_t\cap \{ S_t= (t,s) \}} & = \mathbf{1}_{\{\mathcal{P} \text{ reaches state }S_t=(t,s)\}}. \label{eq:equality_event_pol_1}
	\end{align}
	Note that $x_{t,s}= \Prob(O_t\cap \{ S_t=(t,s) \})$ and $y_{t,s}= \Prob(\overline{O}_t\cap \{ S_t=(t,s) \})$. For $t=1$, then $S_1=(1,1)$ and $\{ \mathcal{P} \text{ reaches state }S_1=(1,1) \}$ occurs with probability $1$. Thus
	\[
	x_{1,1}+ y_{1,1}= 1.
	\]
	For $t> 1$, by the dynamics of the system, the only way that $\mathcal{P}$ reaches state $t$ is by reaching stage $t-1$ and not extending an offer to the $t-1$ candidate or extending an offer but this was turned down. Thus,
	\begin{align}
		\mathbf{1}_{\{\mathcal{P} \text{ reaches state }S_t=(t,s)\}} = \sum_{\sigma=1}^{t-1} \mathbf{1}_{\overline{O}_{t-1} \cap \{ S_{t-1}=(t-1,\sigma)\}\cap \{S_{t}=(t,s)\}} + \mathbf{1}_{O_{t-1}\cap \{S_{t-1}=(t-1,\sigma)\}\cap \overline{D}_{t-1} \cap \{ S_t=(t,s)\}} \label{eq:equality_event_pol_2}
	\end{align}
	Note that
	\begin{align*}
		&\Prob(\overline{O}_{t-1} \cap \{ S_{t-1}=(t-1,\sigma)\}\cap \{S_{t}=(t,s)\} ) \\
		& = \Prob(\overline{O}_{t-1} \cap \{ S_{t-1}=(t-1,\sigma)  \} \mid S_{t}=(t,s)) \Prob(S_t=(t,s)) \\
		& = \Prob(\overline{O}_{t-1} \cap \{ S_{t-1}= (t-1,\sigma) \}) \frac{1}{t} \\
		& = \frac{1}{t}y_{t-1,\sigma}.
	\end{align*}
	Note that we use that $\mathcal{P}$'s action at stage $t-1$ only depends on $S_{t-1}$ and not what is observed in the future. Likewise we obtain
	\begin{align*}
		&\Prob(O_{t-1}\cap \{S_{t-1}=(t-1,\sigma)\}\cap \overline{D}_{t-1} \cap \{ S_t=(t,s)\}) \\
		 & = \Prob(\overline{D}_{t-1})\Prob(O_{t-1}\cap \{S_{t-1}=(t-1,\sigma)\} \cap \{ S_t=(t,s)\}) \\
		& = (1-p) \Prob(O_{t-1}\cap \{S_{t-1}=(t-1,\sigma)\} \cap \{ S_t=(t,s)\}) \\
		& =\frac{1-p}{t} x_{t-1,\sigma}.
	\end{align*} 
	Using the equality between~\eqref{eq:equality_event_pol_1} and~\eqref{eq:equality_event_pol_2} and taking expectation, we obtain
	\begin{align*}
		x_{t,s} + y_{t,s} & = \sum_{\sigma=1}^{t-1} \frac{1}{t} y_{t-1,\sigma} + \frac{1-p}{t} x_{t-1,\sigma}
	\end{align*}
	which shows that $(\x,\y)\in \textsc{Pol}$.
\end{proof}

Conversely

\begin{proposition}\label{prop:tight_points_are_policies}
	Let $(\x,\y)$ be a point in $\textsc{Pol}$. Consider the (randomized) policy $\mathcal{P}$ that in state $(t,s)$ makes an offer to the candidate $t$ with probability $x_{1,1}$ if $t=1$ and
	\[
	\frac{tx_{t,s}}{\sum_{\sigma=1}^{t-1}y_{t-1,\sigma} + (1-p) x_{t-1,\sigma}},
	\]
	if $t>1$. Then $\mathcal{P}$ is a policy for $\SPUA$ such that $x_{t,s}=\Prob(\mathcal{P}\text{ reaches state }(t,s) \text{, selects candidate})$ and $y_{t,s}= \Prob(\mathcal{P}\text{ reaches state }(t,s) \text{, does not select candidate})$ for any $t\in [n]$ and $s\in [t]$.
\end{proposition}

\begin{proof}
	We use the same events $O_t,\overline{O}_t$ and $D_t$ as defined in the previous proof. Thus, we need to show that $x_{t,s}=\Prob(O_t\cap \{ S_t = (t,s)\})$ and $y_{t,s}=\Prob(\overline{O}_t\cap \{ S_t=(t,s) \})$ are the right marginal probabilities. For this, its is enough to show that
	\[
	\Prob(O_t \mid S_t=(t,s) ) = t x_{t,s} \quad \text{ and } \quad \Prob(\overline{O}_t \mid S_t=(t,s) ) = t y_{t,s}
	\]
	for any $t\in [n]$ and for any $s\in [t]$. We prove this by induction in $t\in [n]$. For $t=1$, the result is true by definition of the acceptance probability and the fact that $x_{1,1}+ y_{1,1}=1$. Let us assume the result is true for $t-1$ and let us show it for $t$. First we have
	\begin{align*}
		\Prob(O_t\mid S_t=(t,s) ) & = \Prob(\text{Reach }t, \mathcal{P}(S_t)= \textbf{offer} \mid S_t=(t,s))\\
		& = \Prob(\text{Reach }t \mid S_t=(t,s)) \Prob(\mathcal{P}(S_t) = \textbf{offer}\mid S_t=(t,s)) \\
		& = \Prob(\text{Reach }t \mid S_t=(t,s)) \cdot \frac{t x_{t,s}}{\left( \sum_{\sigma=1}^{t-1} y_{t-1,\sigma} + (1-p) x_{t-1,s}  \right)}
	\end{align*}
	Now, we have
	\begin{align*}
		\Prob(\text{Reach }t\mid S_t=(t,s)) & = \Prob((O_{t-1}\cap \overline{D}_{t-1})\cup \overline{O}_{t-1}\mid S_t=(t,s)) \\
		& = (1-p)\sum_{\sigma=1}^{t-1} \Prob(O_{t-1}\mid S_{t-1}=(t-1,\sigma),S_{t}=(t,s))\Prob(S_{t-1}=(t-1,\sigma)\mid S_{t}=(t,s)) \\
		& \quad \quad + \sum_{\sigma=1}^{t-1} \Prob(\overline{O}_{t-1}\mid S_{t-1}=(t-1,\sigma),S_{t}=(t,s))\Prob(S_{t-1}=(t-1,\sigma)\mid S_{t}=(t,s)) \\
		& = (1-p)\sum_{\sigma=1}^{t-1} \Prob(O_{t-1}\mid S_{t-1}=(t-1,\sigma))\frac{1}{t-1} \\
		& \quad \quad + \sum_{\sigma=1}^{t-1} \Prob(\overline{O}_{t-1}\mid S_{t-1}=(t-1,\sigma))\frac{1}{t-1} \tag{$\mathcal{P}$ only makes decisions at stage $t-1$ based on $S_{t-1}$}\\
		& = (1-p)\sum_{\sigma=1}^{t-1} x_{t-1,\sigma} + y_{t-1,\sigma} \tag{induction}
	\end{align*}
	Note that we used
	\[
	\Prob(S_{t-1}=(t-1,\sigma)\mid S_t=(t,s)) = \frac{\Prob(S_{t}=(t,s)\mid S_{t-1}=(t-1,\sigma)) \Prob(S_{t-1}=(t-1,\sigma))}{\Prob(S_t=(t,s))} = \frac{1}{t-1}.
	\]
	Thus, the induction holds for $\Prob(O_t\mid S_t=(t,s))= t x_{t,s}$. Similarly, for
	\begin{align*}
		\Prob(\overline{O}_t \mid S_t=(t,s)) & = \Prob(\text{Reach }t, \mathcal{P}(S_t)= \textbf{pass}\mid S_t=(t,s)) \\
		& = \Prob(\text{Reach }t\mid S_t=(t,s)) \Prob(\mathcal{P}(S_t)=\textbf{pass}\mid S_{t}=(t,s)) \\
		& = \left( \sum_{\sigma=1}^{t-1} y_{t-1,\sigma} + (1-p)x_{t-1,\sigma}  \right)\left( 1- \frac{t x_{t,s}}{\left( \sum_{\sigma=1}^{t-1} y_{t-1,\sigma} + (1-p) x_{t-1,s}  \right)} \right) \\
		& = t y_{t,s}
	\end{align*}
	where we used the fact that $(\x,\y)\in \textsc{Pol}$.
\end{proof}

\subsubsection{Stage 2: The robust objective}

\begin{proposition}\label{prop:objective_is_linear}
	Let $\mathcal{P}$ be any policy for $\SPUA$ and let $(\x,\y)\in \textsc{Pol}$ be its corresponding vector as in Proposition~\ref{prop:policies_are_points}. Then, for any $k\in [n]$,
	\[
	\Prob(\mathcal{P} \text{ collects a top $k$ candidate}) = \sum_{t=1}^n \sum_{s=1}^t p x_{t,s} \Prob(R_t \leq k \mid r_t=s ).
	\]
\end{proposition}

\begin{proof}
	We use the same events as in the proof of Proposition~\ref{prop:policies_are_points}. Then,
	\begin{align*}
		\Prob(\mathcal{P} \text{ collects a top $k$ candidate}) & =  \sum_{t=1}^n \sum_{s=1}^t \Prob(O_t\cap D_t \cap  \{S_t=(t,s)\} \cap \{ R_t\leq k\}) \\
		& = \sum_{t=1}^n \sum_{s=1}^t p x_{t,s} \Prob(R_t\leq k\mid O_t\cap D_t \cap\{ S_t=(t,s)\} )\\
		& = \sum_{t=1}^n \sum_{s=1}^t p x_{t,s} \Prob(R_t \leq k \mid r_t=s ) 
	\end{align*}
	Note that $R_t$ only depends on $S_t$ and $S_t=(t,s)$ is equivalent to $r_t=s$. 
\end{proof}

We are ready to prove of Theorem~\ref{thm:main_lp}.

\begin{theorem}[Theorem~\ref{thm:main_lp} restated]
	The largest robust ratio $\gamma_n^*$ corresponds to the optimal value of the LP
	
	{\hfil $(LP)_{n,p}$ \hfil \begin{tabular}{rp{0.7\linewidth}}
			$\displaystyle\max_{\x\geq 0}$ & \quad \quad $\gamma$ \\
			s.t. &  {\vspace{-1cm}\begin{align*}
					x_{t,s} & \leq \frac{1}{t} \left( 1 - p \sum_{\tau=1}^{t-1} \sum_{\sigma=1}^\tau x_{\tau,\sigma}  \right) & \forall t\in [n], s\in [t] \\
					\gamma&\leq \frac{p}{1-(1-p)^k} \sum_{t=1}^n \sum_{s=1}^t x_{t,s}\Prob(R_t\leq k \mid r_t=s)& \forall k\in [n] ,
				\end{align*}\vspace{-0.5cm}}
	\end{tabular}}
	
	Moreover, given an optimal solution $(\x^*,\gamma_n^*)$ of $(LP)_{n,p}$, the (randomized) policy $\mathcal{P}^*$ that at state $(t,s)$ makes an offer with probability ${t x_{t,s}^*} / {\left( 1 - p \sum_{\tau=1}^{t-1} \sum_{\sigma=1}^\tau x_{\tau,\sigma}^*  \right)}$ is $\gamma_n^*$-robust.
\end{theorem}

\begin{proof}
	We have
	\begin{align*}
		\gamma_n^* & = \max_{\mathcal{P}} \min_{k\in [n]} \frac{\Prob(\mathcal{P} \text{ collects a candidate with rank}\leq k)}{1-(1-p)^k} \\
		& = \max_{(x,y)\in \textsc{Pol}} \min_{k\in [n]} \frac{ p \sum_{t=1}^n \sum_{s=1}^t x_{t,s}\Prob(R_t\leq k \mid r_t=s)}{1-(1-p)^k} \tag{Propositions~\ref{prop:policies_are_points},~\ref{prop:tight_points_are_policies} and~\ref{prop:objective_is_linear}}
	\end{align*}
	Now note that the function $\gamma:(\x,\y)\mapsto \min_{k\in [n]} \frac{ p \sum_{t=1}^n \sum_{s=1}^t x_{t,s}\Prob(R_t\leq k \mid r_t=s)}{1-(1-p)^k}$ is constant in $\y$. Thus any point $(\x,\y)$ satisfying Constraints~\eqref{const:mdp_dlp_1}-\eqref{const:mpd_dlp_2} has an equivalent point in $(\x',\y')\in\textsc{Pol}$ with $\x'=\x$, $\y'\geq \y$ so all constraints tighten and the objective of $\gamma$ is the same for both points. Thus, $\gamma_n^*$ equals the optimal value of the LP $(P')$:
	
	{\hfil $(P')$ \hfil \begin{tabular}{rp{0.7\linewidth}}
			$\displaystyle\max_{\x\geq 0}$ & \quad \quad $\gamma$ \\
			s.t. &  {\vspace{-1cm}\begin{align*}
					x_{1,1} + y_{1,1} & \leq 1 \\
					x_{t,s} + y_{t,s} & \leq \frac{1}{t}\left( \sum_{\sigma=1}^{t-1} y_{t-1,\sigma} + (1-p) x_{t-1,\sigma} \right) & \forall t\in [n],s\in [t] \\
					\gamma & \leq \frac{ p \sum_{t=1}^n \sum_{s=1}^t x_{t,s}\Prob(R_t\leq k \mid r_t=s)}{1-(1-p)^k}  &\forall k\in [n]
				\end{align*}\vspace{-0.5cm}}
	\end{tabular}}

	where we linearized the objective with the variable $\gamma$. By projecting the feasible region of $(P')$ onto the variables $(\x,\gamma)$ we obtain $(LP)_{n,p}$. This is a routine procedure that can be carried out using Fourier-Motzkin~\citep{schrijver1998theory} but we skip it here for brevity. 
	
	For the second part, we can take an optimal solution $(\x^*,\gamma_n^*)$ and its corresponding point $(\x^*,\y^*)\in \textsc{Pol}$. A routine calculation shows that $1-p\sum_{\tau < t}\sum_{\sigma=1}^\tau x_{\tau,\sigma}^*= \sum_{\sigma=1}^{t-1} y_{t-1,\sigma}^*+ (1-p)x_{t-1,\sigma}^*$. Thus by Proposition~\ref{prop:tight_points_are_policies} we obtain the optimal policy.
\end{proof}

\subsection{Missing proofs from Section~\ref{sec:lp_finite_formulation}: $\gamma_n^*$ is decreasing in $n$}\label{subsec:app:gamma_decreasing_in_n}

\begin{proof}[Proof of Lemma~\ref{lem:non_increasing_gamma}]
	We know that $\gamma_n^*$ equals the value
	
	\resizebox{\columnwidth}{!}{ \hfil \begin{tabular}{rrp{0.7\linewidth}}
			&$\displaystyle\min_{\substack{u : [n]\to \R_+ \\ \sum_i u_i \geq 1}}$ & \quad \quad $v_{(1,1)}$ \\
			(DLP)&s.t. &  {\vspace{-1.3cm}\begin{align*}
					v_{(t,s)} & = \max\left\{ p U_t(s) + \frac{1-p}{t+1} \sum_{\sigma=1}^{t=1} v_{(t+1,\sigma)} , \frac{1}{t+1} \sum_{\sigma=1}^{t=1} v_{(t+1,\sigma)} \right\} & \forall t\in [n], s\in [t] \\
					v_{(n+1,s)}& =0 & \forall s\in [n+1] 
				\end{align*}\vspace{-0.8cm}}
	\end{tabular}}

	where $U_t(s)= \sum_{i=1}^n \frac{u_k}{1-(1-p)^k} \Prob(R_t\in [k]\mid r_t=s)= \sum_{j=1}^n \sum_{k\geq j}\left( \frac{u_k}{1-(1-p)^k}  \right)\Prob(R_t=j\mid r_t=s)$. Thus the utility collected by the policy if it collects a candidate with rank $i$ is $U_i=\sum_{k\geq i} \frac{u_k}{1-(1-p)^k}$.
	Let $u:[n]\to [0,1]$ such that $\sum_{i=1}^n u_i=1$ and extend $u$ to $\widehat{u}:[n+1]\to [0,1]$ by $\widehat{u}_{n+1}=0$ and define $\widehat{U}_t(s)$ accordingly. Consider the optimal policy that solves the program
	\[
	\widehat{v}_{(t,s)} = \max\left\{ p \widehat{U}_t(s) + \frac{1-p}{t+1} \sum_{\sigma=1}^{t+1} \widehat{v}_{(t+1\sigma)} , \frac{1}{t+1} \sum_{\sigma=1}^{t+1} \widehat{v}_{(t+1,\sigma)} \right\}, \forall t\in [n+1], s\in [t]
	\]
	with the boundary condition $\widehat{v}_{(n+2,s)}=0$ for all $s\in [n+2]$. Call this policy $\widehat{\mathcal{P}}$. Note that when policy $\widehat{\mathcal{P}}$ collects a candidate with rank $i$, then it gets a utility of
	\[
	\widehat{U}_i = \sum_{k\geq i} \frac{\widehat{u}_k}{1-(1-p)^k} = \sum_{k\geq i} \frac{u_k}{1-(1-p)^k} = U_i,
	\]
	for $i\leq n$ and $\widehat{U}_{n+1}=0$. By the choice of $\widehat{\mathcal{P}}$, the expected utility collected by $\mathcal{P}$ is $\mathrm{val}(\widehat{\mathcal{P}})= \widehat{v}_{(1,1)}$. We can obtain a policy $\mathcal{P}$ for $n$ elements out of $\widehat{\mathcal{P}}$ by simulating an entry of $n+1$ elements as follows. Policy $\mathcal{P}$ randomly selects a time $t^*\in [n+1]$ and its availability $b$: we set $b=0$ (unavailable) with probability $1-p$ and $b=1$ (available) with probability $p$. Now, on a input of length $n$, the policy $\mathcal{P}$ will squeeze an item of rank $n+1$ in position $t^*$ and it will run the policy $\widehat{\mathcal{P}}$ in this input, simulating appropriately the new partial ranks. That is, before stage $t^*$ policy $\mathcal{P}$ behaves exactly as $\widehat{\mathcal{P}}$ in the original input of $\mathcal{P}$. When the policy leaves the stage $t^*-1$ to transition to stage $t^*$, then the policy $\mathcal{P}$ simulates the simulated candidate $t^*$ (with real rank $n+1$) that $\widehat{\mathcal{P}}$ would have received and does the following: ignores the candidate and moves to stage $t^*$ if the simulated candidate is unavailable ($b=0$) or if $\widehat{\mathcal{P}}((t^*,t^*))=\textbf{pass}$, while if $\widehat{\mathcal{P}}((t^*,t^*)) =\textbf{offer}$ and the simulated candidate accepts ($b=1$) then the policy $\mathcal{P}$ accepts any candidate from that point on. 
	
	Coupling the input of length $n+1$ for $\widehat{\mathcal{P}}$ and the input of length $n$ with the random stage $t^*$ for $\mathcal{P}$, we can see that the utilities collected by $\mathcal{P}$ and $\widehat{\mathcal{P}}$ coincide, i.e., $U(\mathcal{P})= \widehat{\widehat{\mathcal{P}}}=\widehat{v}_{(1,1)}$. Thus the optimal utility $v_{(1,1)}$ collected by a policy for $n$ candidates and utilities given by $U:[n]\to \R_+$, holds $v_{(1,1)} \geq \widehat{v}_{(1,1)}$. Since $\widehat{v}_{(1,1)}\geq \gamma_{n+1}^*$, by minimizing over $u$ we obtain $\gamma_n^*\geq \gamma_{n+1}^*$. 
\end{proof}

\subsection{Missing proofs from Section~\ref{sec:infinite_LP}}\label{sec:app:infinite_lp_approximation}

In this subsection we show that $|\gamma_n^* - \gamma_\infty^*| \leq \mathcal{O}((\log n)^2 / \sqrt{n})$ for fixed $p$ (Lemma~\ref{lem:approximation_inf_LP_finite_LP}). The proof is similar to other infinite approximation of finite models and we require some preliminary results before showing the result. First, we introduce two relaxations, one for $(LP)$ and one for $(CLP)$. We show that the relaxations have values close to their non-relaxed versions. After these preliminaries have been introduced we present the proof of Lemma~\ref{lem:approximation_inf_LP_finite_LP}.

Consider the relaxation of $(LP)_{n,p}$ to the top $q$ candidate constraints:
\begin{align}
	\gamma_{n,q}^*= \max_{x \geq 0}  \quad \gamma & \nonumber\\
	(LP)_{n,p,q} \quad\quad \quad \quad x_{t,s} & \leq \frac{1}{t}\left( 1 - p \sum_{\tau < t}\sum_{s'=1}^\tau x_{\tau, s'}  \right) && \forall t,s \label{const:capped_lp_dynamic} \\
	\gamma & \leq \frac{p}{(1-(1-p)^k)}\sum_{t=1}^n \sum_{s=1}^t x_{t,s} \Prob(R_t\in [k] \mid r_t = s)  && \forall k \in [q] \label{const:capped_lp_multiobj}
\end{align}
Note that $\gamma_n^* \leq  \gamma_{n,q}^*$ since $(LP)_{n,p,q}$ is a relaxation of $(LP)_{n,p}$. The following result gives a bound on $\gamma_n^*$ compared to $\gamma_{n,q}^*$.

\begin{proposition}\label{prop:relaxed_LP_not_bad}
	For any $q\in [n]$, $\gamma_{n}^* \geq \left( 1-(1-p)^q \right) \gamma_{n,q}^*$.
\end{proposition}

\begin{proof}
	Let $(\x,\gamma_{n,q}^*)$ be an optimal solution of $(LP)_{n,p,q}$. Let $f_k = \frac{p}{1-(1-p)^k} \sum_{t=1}^n \sum_{s=1}^t x_{t,s} \Prob(R_t\in [k] \mid r_t=s)$.
	Then $\gamma_{n,q}^* = \min_{k=1,\ldots,q} f_i $. It is enough to show that $f_i \geq \left( \frac{1-(1-p)^q}{1-(1-p)^n}  \right)f_q$ for $i\geq q$ because $\gamma_n^* \geq \min_{i=1,\ldots,n} f_i \geq \left( \frac{1-(1-p)^q}{1-(1-p)^n}  \right) \min_{i=1,\ldots,q} f_i \geq \left( 1-(1-p)^q  \right) \gamma_{n,q}^*$.
	
	For any $j$ we have
	\begin{align*}
		f_{j+1} & = \frac{p}{1-(1-p)^{j+1}} \sum_{t=1}^n \sum_{s=1}^t x_{t,s} \Prob(R_t\in [j+1] \mid r_t=s ) \geq \left(\frac{1-(1-p)^j}{1-(1-p)^{j+1}}\right) f_j.
	\end{align*}
	Thus, iterating this for $j> q$ we get $f_{j}  \geq \left( \frac{1-(1-p)^q}{1-(1-p)^j} \right) f_q$ and we obtain the desired result.
\end{proof}

Likewise we consider the relaxation of $(CLP)_p$ to the top $q$ candidates:

\begin{align}
	\gamma_{\infty, q}^* = \max_{\substack{\alpha:[0,1]\times \N \to [0,1] \\ \gamma \geq 0}}  \quad \gamma & \nonumber \\
	(CLP)_{p,q}\quad \quad \quad \quad \quad \alpha (t, s) & \leq \frac{1}{t}\left( 1 - p \int_0^t \sum_{\sigma \geq 1} \alpha(\tau, \sigma) \,  \mathrm{d}\tau  \right) && \forall t\in [0,1],s\geq 1 \label{const:capped_infinity_dynamic} \\
	\gamma & \leq \frac{p}{(1-(1-p)^k)}\int_0^1  \sum_{s\geq 1} \alpha(t,s) \sum_{\ell=s}^k \binom{\ell-1}{s-1} t^s (1-t)^{\ell-s} \, \mathrm{d}t  && \forall k\in [q] \label{const:capped_infinity_multiobj}
\end{align}

We have $\gamma_{\infty, q}^* \geq \gamma_{\infty}^*$ and we have the approximate converse

\begin{proposition}\label{prop:relaxed_CLP_not_bad}
	For any $q \geq 1$, $\gamma_{\infty}^* \geq \left( 1-(1-p)^q \right) \gamma_{\infty,q}^*$.
\end{proposition}

\begin{proof}
	Let $(\alpha,\gamma_{\infty,q})$ be a feasible solution of $(CLP)_{p,q}$. Let
	\[
	f_k = \frac{p}{1-(1-p)^k} \int_0^1 \sum_{s\geq 1} \alpha(t,s) \sum_{\ell=s}^{k} \binom{\ell-1}{s-1} t^s (1-t)^{\ell -s }\, \mathrm{d}t
	\]
	Assume that $\gamma_{\infty,q} \leq \min_{k\leq q} f_k$. As in the previous proof, we aim to show that $f_i \geq \left( 1 - (1-p)^q \right)f_q$ for $i\geq q$, since this will imply $\gamma_{\infty}^* \geq (1-(1-p)^q) \gamma_{\infty,q}$ for any $(\alpha,\gamma_{\infty,q})$ feasible for $(CLP)_{p,q}$.
	
	Now, for any $j$ we have
	\begin{align*}
		f_{j+1} & = \frac{p}{1-(1-p)^{j+1}} \int_0^1 \sum_{s\geq 1} \alpha(t,s) \sum_{\ell=s}^{j+1} \binom{\ell-1}{s-1} t^s (1-t)^{\ell -s }\, \mathrm{d}t \\
		& \geq \frac{p}{1-(1-p)^{j+1}} \int_0^1 \sum_{s\geq 1} \alpha(t,s) \sum_{\ell=s}^{j} \binom{\ell-1}{s-1} t^s (1-t)^{\ell -s }\, \mathrm{d}t \\
		& = \left( \frac{1-(1-p)^j}{1-(1-p)^{j+1}} \right) f_j.
	\end{align*}
	Iterating the inequality until reaching $q$ we deduce that for any $j\geq q$ we have $f_j \geq \left( \frac{1-(1-p)^q}{1-(1-p)^j}\right) f_q$. From here the result follows.
\end{proof}

\begin{remark}\label{rem:approx_lost}
	If we set $q=(\log n)/p$, both results imply that $\gamma_{n}^*\geq \left( 1-1/n \right)\gamma_{n,q}^*$ and $\gamma_{\infty}^* \geq \left( 1- 1/n \right) \gamma_{\infty,q}^*$. Thus we lose at most $1/n$ by restricting the analysis to the top $q$ candidates.
\end{remark}

\begin{proposition}\label{prop:key_bound_binomial}
	There is $n_0$ such that for $n\geq n_0$, for any $t$ such that $\sqrt{n} \log n \leq t \leq n - \sqrt{n} \log n$,  $\ell \leq {\log (n)} / {p}$ and $\ell \geq s$ it holds that for any $\tau \in \left[ {t}/{n}, {(t+1)}/{n} \right]$ we have
	\[
	1 - \frac{10}{p\sqrt{n}} \leq \frac{\binom{\ell- 1}{s-1} \binom{n-\ell}{t-s}  /\binom{n}{t}}{\binom{\ell-1}{s-1} \tau^s (1-\tau)^{\ell-s} } \leq 1 + \frac{10}{p\sqrt{n}}.
	\]
\end{proposition}

\begin{proof}
	We only need to show that
	\[
	1 - \frac{10}{p\sqrt{n}} \leq \frac{\binom{n-\ell}{t-s} / \binom{n}{t} }{ \tau^s (1-\tau)^{\ell-s} } \leq 1 + \frac{10}{p\sqrt{n}}.
	\]
	We have
	{\small\begin{align*}
		\frac{ \binom{n-\ell}{t-s}  }{\binom{n}{t}} & = \frac{t!}{(t-s)!} \frac{1}{n!} \frac{(n-\ell)! (n-t)!}{(n-t - (\ell-s))!}\\
		& = \left(\frac{t\cdot (t-1) \cdots (t-s+1)}{n\cdot (n-1) \cdots (n-s+1)} \right)\left( \frac{(n-t)(n-t-1)\cdots (n-t-(\ell-s)+1)}{(n-s)(n-s-1)\cdots (n-s - (\ell-s)+1)} \right) \\
		& = \underbrace{\left(\frac{t}{n}\right)^s \left( 1- \frac{t}{n}  \right)^{\ell-s}}_A \underbrace{\left( \frac{1\cdot\left(1- \frac{1}{t} \right) \cdots \left( 1- \frac{(s-1)}{t} \right) }{1\cdot\left(1- \frac{1}{n} \right) \cdots \left( 1- \frac{(s-1)}{n} \right)} \right)}_B\underbrace{\left( \frac{ 1\cdot\left(1- \frac{1}{n-t} \right) \cdots \left( 1- \frac{(\ell-s)-1}{n-t} \right) }{\left( 1- \frac{s}{n} \right)\cdot\left(1- \frac{(s+1)}{n} \right) \cdots \left( 1- \frac{(s+(\ell-s)-1)}{n} \right)} \right)}_{C}
	\end{align*}}
	We bound terms $A, B$ and $C$ separately. Since $s\leq \ell $ and we are assuming that $\ell \leq (\log n)/p$ and $t\geq \sqrt{n}\log n$ for $n$ large, then we will implicitly use that $s,\ell \leq \min \{t/2, n/2\}$.
	
	\begin{claim*} It holds
		$\left( 1 - {4} / {(p\sqrt{n})} \right)\tau^s (1-\tau)^{\ell-s}\leq A = \left({t} / {n}\right)^s \left( 1- {t}/{n}  \right)^{\ell-s}\leq \left( 1 + {4} / {(p\sqrt{n})} \right)\tau^s (1-\tau)^{\ell-s}$.
	\end{claim*}

	\begin{proof}
		For the upper bound we have
		\begin{align*}
			\left(\frac{t}{n}\right)^s \left( 1- \frac{t}{n}  \right)^{\ell-s} & \leq \tau^s \left(  1- \tau + \frac{1}{n} \right)^{\ell-s} \tag{$\tau\in [t/n,(t+1)/n]$}\\
			& = \tau^s (1-\tau)^{\ell-s} \left(  1+ \frac{1}{(1-\tau)n} \right)^{\ell-s} \\
			& \leq \tau^s (1-\tau)^{\ell-s} e^{{(\ell-s)}/{\left((1-\tau)n\right)}}\\
			& \leq \tau^s (1-\tau)^{\ell-s} e^{{\ell}/{(n-t-1)}}\\
			& \leq \tau^s (1-\tau)^{\ell-s} \left( 1 + 2\frac{\ell}{n-t-1}  \right) \tag{Using $e^x \leq 1+2x$ for $x\in [0,1]$}
		\end{align*}
		The upper bound now follows by using the information over $\ell$ and $t$ and that $\left( 1 + 2{\ell}/{(n-t-1)}  \right)\leq 1+2 {(\log n)}{p^{-1}(\sqrt{n}\log n - 1)^{-1}}\leq 1+ {4}/{(p\sqrt{n})}$ for $n$ large.
		
		For the lower bound we have
		\begin{align*}
			\left( \frac{t}{n} \right)^s \left( 1- \frac{t}{n} \right)^{\ell-s} & \geq \left( \tau - \frac{1}{n}  \right)^s \left( 1- \tau  \right)^{\ell-s} \\
			& = \tau^s(1-\tau)^{\ell-s} \left( 1- \frac{1}{\tau n} \right)^s \\
			& \geq \tau^s (1-\tau)^{\ell-s} e^{-\frac{s}{\tau n -1}} \\
			& \geq \tau^s (1-\tau)^{\ell-s} \left(  1 - {s}/{(\tau n - 1)} \right)
		\end{align*}
		Since ${s} / {(\tau n-1)} \leq {\log (n)} / {(p (t-1))} \leq {2}/{(p\sqrt{n})}$ for $n$ large, the lower bound follows.
	\end{proof}

	\begin{claim*}
		We have $1- 2 {s^2} / {t} \leq B \leq 1 + 2 {s^2} / {n}$
	\end{claim*}

	\begin{proof}
		For the upper bound we upper bound the denominator
		\begin{align*}
			\left( \frac{1\cdot\left(1- \frac{1}{t} \right) \cdots \left( 1- \frac{(s-1)}{t} \right) }{1\cdot\left(1- \frac{1}{n} \right) \cdots \left( 1- \frac{(s-1)}{n} \right)} \right) & \leq \frac{1}{1\cdot\left(1- \frac{1}{n} \right) \cdots \left( 1- \frac{(s-1)}{n} \right)} \\
			&\leq   e^{\sum_{k=1}^{s-1} {k}/{(n-k)}} \\
			& \leq e^{{s^2}/{n}} \tag{Function $x \mapsto {x}/{(n-x)}$ is increasing}\\
			& \leq 1 + 2 \frac{s^2}{n}
		\end{align*}
		For the lower bound we lower bound the numerator:
		\begin{align*}
			\left( \frac{1\cdot\left(1- \frac{1}{t} \right) \cdots \left( 1- \frac{(s-1)}{t} \right) }{1\cdot\left(1- \frac{1}{n} \right) \cdots \left( 1- \frac{(s-1)}{n} \right)} \right) &  \geq 1\cdot\left(1- \frac{1}{t} \right) \cdots \left( 1- \frac{(s-1)}{t} \right) \\
			&\geq e^{- \sum_{k=1}^{s-1} {k}/{(t-k)}} \tag{Using $1-{k}/{t} = {\left(1+ {k}/{(t-k)}\right)^{-1}} \geq e^{-{k}/{(t-k)}}$} \\
			& \geq 1- \frac{s^2}{t-s} \geq 1- 2 \frac{s^2}{t}.  
		\end{align*}
	\end{proof}

	\begin{claim*}
		We have $1- 2 {\ell^2}/{(n-t)}\leq C \leq 1+ 2 {\ell^2}/{n}$
	\end{claim*}

	\begin{proof}
		Similar to the previous claim, we bound denominator for an upper bound and numerator for a lower bound.
		\begin{align*}
			\left( \frac{ 1\cdot\left(1- \frac{1}{n-t} \right) \cdots \left( 1- \frac{(\ell-s)-1}{n-t} \right) }{\left( 1- \frac{s}{n} \right)\cdot\left(1- \frac{(s+1)}{n} \right) \cdots \left( 1- \frac{(s+(\ell-s)-1)}{n} \right)} \right) & \leq \frac{1}{\left( 1- \frac{s}{n} \right)\cdot\left(1- \frac{(s+1)}{n} \right) \cdots \left( 1- \frac{(s+(\ell-s)-1)}{n} \right)} \\
			& \leq e^{\sum_{k=0}^{\ell-s-1} {(k+s)}/{(n-k)}}  \leq 1 + 2 \frac{\ell^2}{n},
		\end{align*}
		and
		\begin{align*}
			\left( \frac{ 1\cdot\left(1- \frac{1}{n-t} \right) \cdots \left( 1- \frac{(\ell-s)-1}{n-t} \right) }{\left( 1- \frac{s}{n} \right)\cdot\left(1- \frac{(s+1)}{n} \right) \cdots \left( 1- \frac{(s+(\ell-s)-1)}{n} \right)} \right) & \geq 1\cdot\left(1- \frac{1}{n-t} \right) \cdots \left( 1- \frac{(\ell-s)-1}{n-t} \right) \\
			& \geq e^{- \sum_{k=0}^{\ell-s -1} {k}/{(n-t-k)}}  \geq 1 - 2\frac{\ell^2}{n-t}.
		\end{align*}
	\end{proof}
	We can now upper bound $ABC$ as
	\begin{align*}
		ABC &\leq \tau^s (1-\tau)^{\ell-s} \left(  1 + \frac{4}{p\sqrt{n}} \right)\left( 1+ 2 \frac{s^2}{n} \right)\left( 1+ 2\frac{\ell^2}{n}\right) \\
		& \leq  \tau^s (1-\tau)^{\ell-s} \left(  1 + \frac{4}{p\sqrt{n}}  \right) \left( 1 + 2 \frac{(\log n)^2}{p^2 n}  \right)^2 \tag{Using $t\leq n - \sqrt{n}\log n$ and $s\leq \ell \leq (\log n)/p$} \\
		& \leq \tau^s (1-\tau)^{\ell-s}\left(  1 + \frac{4}{p\sqrt{n}} \right)\left( 1 + 6 \frac{(\log n)^2}{p^2 n}  \right) \tag{Using $(1+x)^2 \leq 1 + 3x$ if $x\in [0,1]$} \\
		& \leq \tau^s (1-\tau)^{\ell-s} \left( 1 + \frac{10}{p \sqrt{n}}  \right).
	\end{align*}
	Recall that we are assuming $p$ constant and $n$ large, thus the dominating term is $1/\sqrt{n}$. Similarly, we can lower bound $ABC$ as
	\begin{align*}
		ABC & \geq \tau^s(1-\tau)^{\ell-s} \left( 1 - \frac{4}{p\sqrt{n}}   \right)\left(1 - 2 \frac{s^2}{t}\right) \left(  1 - 2 \frac{\ell^2}{n-t}\right)\\
		& \geq \tau^s (1-\tau)^{\ell-s} \left( 1 - \frac{4}{p\sqrt{n}}   \right) \left( 1- 2 \frac{(\log n)^2}{p^2 t} \right)\left(  1 - 2 \frac{(\log n)^2}{p^2(n-t)} \right) \\
		& \geq \tau^s (1-\tau)^{\ell-s}\left( 1 - \frac{4}{p\sqrt{n}}   \right) \left( 1 - 2 \frac{(\log n)}{p^2 \sqrt{n}} \right)^2 
		\geq \tau^s (1-\tau)^{\ell-s} \left( 1 - \frac{10}{p\sqrt{n}}  \right).
	\end{align*}
\end{proof}

\begin{proof}[Proof of Lemma~\ref{lem:approximation_inf_LP_finite_LP}] We are going to show $|\gamma_n^* - \gamma_{\infty}^*|\leq \mathcal{O}((\log n)^2/\sqrt{n})$. Since we can only guarantee good approximation of the binomial terms in Proposition~\ref{prop:key_bound_binomial} for $\ell \leq (\log n)/p$, we need to restrict our analysis to $\gamma_{n,q}^*$ and $\gamma_{\infty,q}^*$ for $q=(\log n)/p$. This is enough since these values are withing $1/n$ of $\gamma_n^*$ and $\gamma_{\infty}^*$ respectively due to Propositions~\ref{prop:relaxed_LP_not_bad} and~\ref{prop:relaxed_CLP_not_bad} (see Remark~\ref{rem:approx_lost}). 
	
	Before proceeding, we give two technical results that allow us to control an error for values of $t$ not considered by Proposition~\ref{prop:key_bound_binomial}. The deduction is a routine calculation and it is skipped for brevity.

	\begin{claim*}\label{claim:small_t_values}
		For any $x$ feasible for Constraints~\eqref{const:capped_lp_dynamic} and such that $x_{t,s}=0$ for $s> q$, we have for $k\leq q$
		\begin{itemize}
			\item $\sum_{t=1}^{\sqrt{n}\log n} \sum_{s=1}^t x_{t,s} \Prob(R_t\in [k] \mid r_t=s) \leq 10 {(\log n)^2} /{(p\sqrt{n})} $.
			\item $\sum_{t=n-\sqrt{n}\log n}^{n} \sum_{s=1}^t x_{t,s} \Prob(R_t\in [k] \mid r_t=s) \leq 10 {(\log n)^2} /{(p\sqrt{n})} $.
		\end{itemize}
	\end{claim*}

	\begin{claim*}\label{claim:small_t_values_CLP}
		For any $\alpha$ feasible for Constraints~\eqref{const:capped_infinity_dynamic}, we have for $k\leq q$
		\begin{itemize}
			\item $\int_{1-(\log n)/\sqrt{n}}^1 \sum_{s=1}^k \alpha(\tau, s) \sum_{\ell=s}^k \binom{\ell-1}{s-1} \tau^s(1-\tau)^{\ell-s} \mathrm{d}\tau \leq {(\log n)^2} /{(p\sqrt{n})}$.
			\item $\int_{0}^{(\log n)/\sqrt{n}} \sum_{s=1}^k \alpha(\tau, s) \sum_{\ell=s}^k \binom{\ell-1}{s-1} \tau^s(1-\tau)^{\ell-s} \mathrm{d}\tau \leq {(\log n)^2} /{(p\sqrt{n})}$.
		\end{itemize}
	\end{claim*}

	First, we show that $\gamma_{n,q}^* \geq \gamma_{\infty}^* - 40 (\log n)^2/(p\sqrt{n})$. Let $(\alpha,\gamma)$ be a feasible solution of the continuous LP $(CLP)_p$. We construct a solution of the $(LP)_{n,p,q}$ as follows. Define
	\begin{align*}
		x_{t,s} &= \frac{t-(1-p)}{t} \int_{{(t-1)}/{n}}^{{t}/{n}} \alpha(\tau ,s) \, \mathrm{d} \tau && \forall t\in [n], \forall s\in [t],
	\end{align*}
	and $\gamma_{n,q} = \min_{k\leq q} \frac{p}{1-(1-p)^k} \sum_{t=1}^n \sum_{s=1}^t x_{t,s} \Prob(R_t\in [k] \mid r_t = s)$. Let us show that $(\x,\gamma_{n,q})$ is feasible for $(LP)_{n,q}$, i.e., it satisfies Constraints~\eqref{const:capped_lp_dynamic}-\eqref{const:capped_lp_multiobj}. First, for $\tau\in\left[ {(t-1)}/{n}, {t} / {n} \right]$ we have
	\begin{align*}
		\tau \alpha(\tau,s) + p\int_{{(t-1)}/{n}}^\tau \alpha(\tau',s) \, \mathrm{d}\tau'  & \leq 1- p\int_{0}^\tau \sum_{\sigma\geq 1} \alpha(\tau', \sigma) \,\mathrm{d} \tau' + p\int_{{(t-1)}/{n}}^\tau \alpha(\tau',s) \, \mathrm{d}\tau' \\
		& \leq 1 - p \int_0^{{(t-1)}/{n}} \sum_{\sigma \geq 1} \alpha(\tau',\sigma) \, \mathrm{d}\tau'  \leq 1 - p\sum_{\tau'=1}^{t-1} \sum_{\sigma=1}^{\tau'} x_{\tau',\sigma}.
	\end{align*}
	We now integrate in $\left[ {(t-1)}/{n}, {t}/{n} \right]$ on both sides of the inequality. After integration, the RHS equals $\left(  1 - p \sum_{\tau=1}^t \sum_{\sigma=1}^\tau x_{\tau,\sigma}  \right) / n$. On the LHS we obtain,
	\begin{align*}
		\int_{{(t-1)}/{n}}^{{t}/{n}} \left( \tau \alpha(\tau,s) + p \int_{{(t-1)}/{n}}^\tau \alpha(\tau',s) \mathrm{d} \tau'   \right) \, \mathrm{d}\tau & = \frac{t}{n}\int_{{{(t-1)}}/{n}}^{{t}/{n}} \alpha(\tau,s) \mathrm{d}\tau - (1-p) \int_{{(t-1)}/{n}}^{{t}/{n}} \left( \frac{t}{n} - \tau \right)\alpha(\tau,s)\mathrm{d}\tau \\
		& \geq \frac{(t-(1-p))}{n} \int_{{(t-1)}/{n}}^{{t}/{n}} \alpha(\tau,s)\mathrm{d}\tau \tag{Using $t/n-\tau \leq 1/n$} \\ 
		&= \frac{t}{n} x_{t,s}.
	\end{align*}
	Thus Constraints~\eqref{const:capped_lp_dynamic} hold. By definition of $\gamma_{n,q}$, Constraints~\eqref{const:capped_lp_multiobj} also hold.
	
	Now, note that for $t\geq \sqrt{n}\log n$ we have
	$$x_{t,s} \geq \left( 1 - \frac{1}{\sqrt{n}\log n} \right) \int_{{(t-1)}/{n}}^{{t}/{n}} \alpha(\tau,s)\mathrm{d}\tau.$$
	Then, 
	{\small \begin{align*}
		\gamma_{n,q} &  = \min_{k\leq q} \frac{p}{1-(1-p)^k} \sum_{t=1}^n \sum_{s=1}^t x_{t,s} \sum_{\ell=s}^{k\wedge(n-t+s)} \frac{\binom{\ell-1}{s-1} \binom{n-\ell}{t-s}}{\binom{n}{t}} \tag{Definition of $\Prob(R_t\in [k]\mid r_t=s)$} \\
		& \geq \min_{k\leq q} \frac{p}{1-(1-p)^k} \sum_{t=\sqrt{n}\log n}^{n-\sqrt{n}\log n} \sum_{s= 1}^t \int_{{(t-1)}/{n}}^{{t}/{n}} \alpha(\tau, s) \, \mathrm{d} \tau  \sum_{\ell=s}^{k\wedge(n-t+s)} \frac{\binom{\ell-1}{s-1} \binom{n-\ell}{t-s}}{\binom{n}{t}} \left( 1- \frac{1}{\sqrt{n}\log n} \right) \\
		& \geq \min_{k\leq q} \frac{p}{1-(1-p)^k} \sum_{t=\sqrt{n}\log n}^{n-\sqrt{n}\log n} \sum_{s= 1}^k \int_{{(t-1)}/{n}}^{{t}/{n}} \alpha(\tau, s)    \sum_{\ell=s}^{k} \binom{\ell-1}{s-1}\tau^s (1-\tau)^{\ell-s}  \, \mathrm{d}\tau \left( 1 - \frac{20}{p\sqrt{n}}  \right) \tag{Since $n-t+s \geq \sqrt{n}\log n \geq k$ and $t \leq k$ and using Proposition~\ref{prop:key_bound_binomial}} \\
		& = \min_{k\leq q} \frac{p}{1-(1-p)^k} \sum_{t=\sqrt{n}\log n}^{n-\sqrt{n}\log n} \int_{{(t-1)}/{n}}^{{t}/{n}} \sum_{\ell = 1}^k \sum_{s=1}^{\ell} \alpha(\tau, s)  \binom{\ell-1}{s-1}\tau^s (1-\tau)^{\ell-s}  \, \mathrm{d}\tau \left( 1 - \frac{20}{p\sqrt{n}}  \right) \\
		& \geq \min_{k\leq q} \left(\frac{p}{1-(1-p)^k} \int_{0}^{1} \sum_{\ell = 1}^k \sum_{s=1}^{\ell} \alpha(\tau, s)  \binom{\ell-1}{s-1}\tau^s (1-\tau)^{\ell-s}  \, \mathrm{d}\tau - 2\frac{(\log n)^2}{p\sqrt{n}}\right) \left( 1 - \frac{20}{p\sqrt{n}}  \right) \tag{Claim~\ref{claim:small_t_values_CLP}} \\
		& \geq \left( \gamma - 2\frac{(\log n)^2}{p\sqrt{n}}   \right) \left( 1- \frac{10}{p\sqrt{n}}  \right) \geq  \gamma - \frac{20}{p\sqrt{n}} - 2 \frac{\log n}{\sqrt{n}}.
	\end{align*}}
	\modif{With this, we have proved that $(x,\gamma_{n,q})$ is a feasible solution of $(LP)_{n,p,q}$ with an objective value $\gamma_{n,q}$ at least $\gamma - {20}/{p\sqrt{n}} - 2 {\log n}/{\sqrt{n}}$. Hence, the optimal value of $(LP)_{n,p,q}$, $\gamma_{n,q}^*$ is at least $\gamma - {20}/{p\sqrt{n}} - 2 {\log n}/{\sqrt{n}}$. Since, $(\alpha,\gamma)$ is any feasible solution of $(CLP)_{p,q}$, and $\gamma_{\infty,q}^*\geq \gamma_\infty^*$, we obtain $\gamma_{n,q}^* \geq \gamma_\infty^* - 40 {\log (n)} /{\left(p\sqrt{n}\right)}$ for $n$ large.  }
		
	
	Now, we show that $\gamma_{\infty,q}^* \geq \gamma_{n,q}^* - 40 (\log n)^2/(p \sqrt{n})$. Let $(x,\gamma_{n,q})$ be a solution of $(LP)_{n,p,q}$. Let us construct a solution of $(CLP)_{p,q}$. Note that we can assume $x_{t,s}=0$ for $s > q$ as $(LP)_{n,p,q}$ does not improve its objective function by allocating any mass to these variables. Consider $\alpha$ defined as follows: for $\tau\in [0,1]$ let
	\[
	\alpha(\tau,s) = \begin{cases}
		n x_{t,s}\left( 1- {\log (n)}/{\sqrt{n}} \right) & t= \lceil \tau n \rceil \geq \sqrt{n}, s \leq \min\{ t, \log n/p  \} \\
		0 & t= \lceil \tau n \rceil < \sqrt{n} \text{ or } s  > \min\{ t, \log n/p  \} \\
	\end{cases}
	\]
	Let $\gamma_{\infty,q} = \min_{k\leq q} \frac{p}{1-(1-p)^k} \int_{0}^1 \sum_{s\geq 1}\alpha(\tau,s) \sum_{\ell=s}^{k} \binom{\ell-1}{s-1} \tau^s(1-\tau)^{\ell-s}\, \mathrm{d}\tau $. We show first that $(\alpha,\gamma_{\infty,q})$ is feasible for $(CLP)_{p,q}$, and for this it is enough to show that $\alpha$ holds Constraints~\eqref{const:capped_lp_dynamic}. For $\tau < 1/\sqrt{n}$ we have $\alpha(\tau, s)=0$ for any $s$, thus Constraint~\eqref{const:capped_infinity_dynamic} is satisfied in this case. Let us verify that for $\tau\geq 1/\sqrt{n}$ the constraint also holds. Let $t= \lceil \tau n \rceil$ and $s\geq 1$. Then
	\begin{align*}
		\tau \alpha(\tau, s) + p \int_0^\tau \sum_{\sigma \geq 1} \alpha(\tau', \sigma) \, \mathrm{d} \tau' & \leq \tau \alpha(\tau, s) + p \sum_{t'=1}^{t-1} \int_{\frac{t'-1}{n}}^{\frac{t'}{n}}\sum_{s=1}^{t'} \alpha(\tau',s) \, \mathrm{d}\tau' +  p\int_{\frac{t-1}{n}}^\tau \sum_{\sigma=1}^{\frac{\log n}{p}} \alpha(\tau',\sigma) \, \mathrm{d}\tau' \\
		& \leq \left( 1 - \frac{\log n}{\sqrt{n}}  \right) \left( t x_{t,s} +p \sum_{t'=1}^{t-1}\sum_{s=1}^{t'} x_{t',s} + p \frac{\log n}{p t}  \right) \tag{Since $x_{t,s}\leq \frac{1}{t}$ always} \\
		& \leq \left( 1- \frac{\log n}{\sqrt{n}} \right) \left(  1 + \frac{\log n}{t}  \right) \\
		& \leq \left( 1 - \frac{\log n}{\sqrt{n}} \right) \left( 1 + \frac{\log n}{\sqrt{n}} \right).  \tag{$t \geq \sqrt{n}$}
	\end{align*}
	The last term is $<1$. Thus $(\alpha, \gamma_{\infty,q})$ is feasible for $(CLP)_{p,q}$. Now, 
	{\small\begin{align*}
		\gamma_{\infty,q} & = \min_{k\leq q} \frac{p}{1-(1-p)^k} \int_{0}^1 \sum_{s\geq 1} \alpha(\tau,s) \sum_{\ell=s}^{k} \binom{\ell-1}{s-1} \tau^s (1-\tau)^{\ell-s} \, \mathrm{d}\tau \\
		& \geq \min_{k\leq q} \frac{p}{1-(1-p)^k} \sum_{t=\sqrt{n}\log n}^{n- \sqrt{n}\log n} \int_{\frac{t-1}{n}}^{\frac{t}{n}} \sum_{s\geq 1} \alpha(\tau,s) \sum_{\ell=s}^{k} \binom{\ell-1}{s-1} \tau^s (1-\tau)^{\ell-s} \, \mathrm{d}\tau \\
		& \geq \min_{k\leq q} \frac{p}{1-(1-p)^k} \sum_{t=\sqrt{n}\log n}^{n- \sqrt{n}\log n} \int_{\frac{t-1}{n}}^{\frac{t}{n}} \sum_{s\geq 1} \alpha(\tau,s) \sum_{\ell=s}^{k} \frac{\binom{\ell-1}{s-1} \binom{n-\ell}{t-s}}{\binom{n}{t}} \, \mathrm{d}\tau \left( 1- \frac{10}{p\sqrt{n}}  \right) \tag{Proposition~\ref{prop:key_bound_binomial}} \\
		&\geq \min_{k\leq q} \frac{p}{1-(1-p)^k} \sum_{t=\sqrt{n}\log n}^{n- \sqrt{n}\log n} \sum_{s=1}^t x_{t,s} \sum_{\ell=s}^{k} \frac{\binom{\ell-1}{s-1} \binom{n-\ell}{t-s}}{\binom{n}{t}} \, \mathrm{d}\tau \left( 1- \frac{\log n}{\sqrt{n}} \right)\left( 1- \frac{10}{p\sqrt{n}}  \right)\\
		&\geq \left(\min_{k\leq q} \frac{p}{1-(1-p)^k} \sum_{t=1}^{n} \sum_{s=1}^t x_{t,s} \sum_{\ell=s}^{k} \frac{\binom{\ell-1}{s-1} \binom{n-\ell}{t-s}}{\binom{n}{t}} \, \mathrm{d}\tau - 20 \frac{(\log n)^2}{p\sqrt{n}} \right) \left( 1- \frac{\log n}{\sqrt{n}} \right)\left( 1- \frac{10}{p\sqrt{n}}  \right) \tag{Claim~\ref{claim:small_t_values}} \\
		& \geq \gamma_{n,q} - 40 \frac{(\log n)^2}{p \sqrt{n}}.
	\end{align*}}

	\modif{With this, we have formed a feasible solution $(\alpha,\gamma_{\infty,q})$ of $(CLP)_{p,q}$. Hence, $\gamma_{\infty,q}^*\geq \gamma_{n,q}-40(\log n)^2/p\sqrt{n}$. Since $(x,\gamma_{n,q})$ is any feasible solution of $(LP)_{n,p,q}$, we can optimize over $(x,\gamma_{n,q})$ and obtain the inequality $\gamma_{\infty,q}^* \geq \gamma_{n,q}^* - 40 (\log n)^2/(p\sqrt{n})$. }
	
	\modif{Using Propositions~\ref{prop:relaxed_LP_not_bad} and~\ref{prop:relaxed_CLP_not_bad} we can conclude that, for $n$ large, $\gamma_{\infty}^* - 50 {(\log n)^2} /{\left(p\sqrt{n}\right)} \leq \gamma_n^* \leq \gamma_{\infty}^* + 50 {(\log n)^2} /{\left(p\sqrt{n}\right)}$, where the additional constant factors appear as a byproduct of choosing $q=\log n/p$ in both propositions.}

\end{proof}

\subsection{Missing proofs from Section~\ref{sec:upper_bound}}\label{sec:app:upper_bounds}


The following result is the reduction from $\SPUA$ to i.i.d.\ prophet inequality propblem

\begin{lemma}\label{lem:reduction_sec_to_proph_1}
	Then there is an algorithm $\mathcal{A}'$ for the i.i.d.\ prophet inequality problem that for any $\varepsilon,\delta>0$ satisfying
	\[
	(1+\varepsilon)p < 1, \qquad n\geq \frac{2}{p\varepsilon^2}\log \left(\frac{2}{\delta}\right), \qquad m=\lfloor (1+\varepsilon)p n \rfloor,
	\]
	ensures
	\[
	\E\left[\mathrm{Val}(\mathcal{A}')\right] + \delta \geq \gamma (1- 4\varepsilon - \delta)\E\left[\max_{i\leq m} X_i\right],
	\]
	for any $X_1,\ldots,X_m$ sequence of i.i.d.\ random variables with support in $[0,1]$, where $\mathrm{Val}(\mathcal{A}')$ is the profit obtained by $\mathcal{A}'$ from the sequence of values $X_1,\ldots,X_m$ in the prophet problem.
\end{lemma}

\begin{proof}
	The input of the i.i.d.\ prophet inequality problem corresponds to a known distribution $\mathcal{D}$ with support in $[0,1]$. The DM sequentially accesses at most $m$ samples from $\mathcal{D}$ and upon observing one of these values, she has to decide irrevocably if to take it and stop the process or continue. We are going to use $\mathcal{A}$ to design a strategy for the prophet problem. We assume that the samples from $\mathcal{D}$ are all distinct. Indeed, we can add some small Gaussian noise to the distribution and consider a continuous distribution $\mathcal{D}'$ instead. 
	
	Note that $\mathcal{A}$ runs on an input of size $n$ where a fraction $p$ of the candidates accept an offer. We interpret $pn \approx m$ as the set of samples for the prophet inequality problem, while the remaining $(1-p)n$ items are used as additional information for the algorithm. By concentration bounds, we are going to argue that we only need to run $\mathcal{A}$ in at most $(1+\varepsilon)pn$ positive samples.
	
	Formally, we proceed as follows. Fix $n$ and $\varepsilon>0$ and consider the algorithm $\mathcal{B}$ that receives an online input of $n$ numbers $x_1,\ldots,x_n$. The algorithm flips $n$ coins with probability of heads $p$ and marks item $i$ as available if the corresponding coins that turn out heads. Algorithm $\mathcal{B}$ feeds algorithm $\mathcal{A}$ with the partial rankings given by the ordering given by $x_1,\ldots,x_n$. If $\mathcal{A}$ selects a candidate but the candidate is marked as unavailable, then $\mathcal{B}$ moves to the next item. If $\mathcal{A}$ selects a candidate $i$ and it is marked as available, then the process ends with $\mathcal{B}$ collecting the value $x_i$. Let us denote by $\mathrm{Val}(\mathcal{B},x_1,\ldots,x_n)$ the value collected by $\mathcal{B}$ in the online input $x_1,\ldots,x_n$. Then we have the following claim.
	
	\begin{claim*}
		$\E_{\substack{X_1,\ldots,X_n\\ S}}\left[ \mathrm{Val}(\mathcal{B},X_1,\ldots,X_n) \right]\geq \gamma \E_{\substack{X_1,\ldots,X_n\\S}}$$\left[ \max_{i\in S}X_i \right]$, where $S$ is the random set of items marked as available and $X_1,\ldots,X_n$ are $n$ i.i.d.\ random variables with common distribution $\mathcal{D}$.
	\end{claim*}

	\begin{proof}
		Fix $x_1,\ldots,x_n$ points in the support of $\mathcal{D}$. Then, a simple application of Proposition~\ref{prop:utility_characterization} shows
		\begin{align*}
			\frac{\E_{S,\pi}\left[\mathrm{Val}(\mathcal{B},x_{\pi(1)},\ldots,x_{\pi(n)})\right]}{\E_{\substack{S}}\left[\max_{i\in S}x_i  \right]}
			& \geq \gamma
		\end{align*}
		Note that we need to feed $\mathcal{B}$ with all permutations of $x_1,\ldots,x_n$ in order to obtain the guarantee of $\mathcal{A}$. From here, we obtain $\E_{\mathcal{S},\pi}\left[\mathrm{Val}(\mathcal{B},x_\pi(1),\ldots,x_\pi(n))\right] \geq \gamma \E_{\substack{\mathcal{S}}}\left[\max_{i\in S}x_i  \right]$ and the conclusion follows by taking expectation in $X_1=x_1,\ldots,X_n=x_n$.
	\end{proof}
	
	For ease of notation, we will refer by $\mathrm{Val}(\cdot)$ to $\mathrm{Val}(\cdot,X_1,\ldots,X_n)$. We modify slightly $\mathcal{B}$. Consider $\mathcal{B}'$ that runs normally $\mathcal{B}$ if $|S|\leq (1+\varepsilon)pn$ or simply return $0$ value if $|S|> (1+\varepsilon)pn$. Then, we have

	\begin{claim*}
		Let $\varepsilon, \delta >0$. For $n\geq {2}\log \left( {2}/{\delta} \right)/({p\varepsilon^{2}})$ we have $\E_{\substack{X_1,\ldots,X_n\\ \mathcal{S}}}\left[ \max_{i\in \mathcal{S}} X_i  \right] \geq \left( 1 - \delta \right) \E\left[  \max\limits_{i \leq (1-\varepsilon) pn } X_i   \right]$ and $\E\left[ \mathrm{Val}(\mathcal{B}')  \right]+ \delta \geq \E\left[ \mathrm{Val}(\mathcal{B}) \right]$.
	\end{claim*}

	\begin{proof}
		Using standard Chernoff concentration bounds (see for instance~\citep{boucheron2013concentration}) we get			$\Prob_{\mathcal{S}}\left( \left| |\mathcal{S}| - pn    \right| \geq \varepsilon pn  \right) \leq 2 e^{-{pn \varepsilon^2}/{2}} = \delta.$ 		Hence, for $n \geq {2}\log \left( {2}/{\delta} \right)/({p\varepsilon^{2}})$, we can guarantee that
		\[
		\E_{\substack{X_1,\ldots,X_n\\ \mathcal{S}}}\left[ \max_{i\in \mathcal{S}} X_i  \right] \geq \left( 1 - \delta \right) \E\left[  \max_{i \leq (1-\varepsilon) pn } X_i   \right].
		\]
		For the second part we have $\E\left[ \mathrm{Val}(\mathcal{B})  \right] \leq \delta + \E\left[  \mathrm{Val}(\mathcal{B}) \mid |S| \leq (1+\varepsilon)pn \right] = \delta + \E\left[ \mathrm{Val}(\mathcal{B}')  \right]$.
	\end{proof}

	\begin{claim*}
		For any $\varepsilon>0$ we have $\E\left[ \max\limits_{i \leq (1-\varepsilon) pn } X_i  \right]  \geq (1-\varepsilon)^2 \E\left[ \max\limits_{i \leq (1+\varepsilon) pn } X_i  \right]$.
	\end{claim*}

	\begin{proof}
		Since $\Prob\left(  \max_{i\leq k} X_i \leq x   \right)= \Prob(X_1 \leq x)^k$, then we have
		\begin{align*}
			\frac{\E\left[ \max\limits_{i \leq (1-\varepsilon) pn } X_i  \right]}{\E\left[ \max\limits_{i \leq (1+\varepsilon) pn } X_i  \right]} & = \frac{\int_0^\infty \left( 1- \Prob(X_1 \leq x)^{pn(1-\varepsilon)}  \right)\, \mathrm{d}x}{\int_0^\infty \left( 1- \Prob(X_1 \leq x)^{pn(1+\varepsilon)}  \right)\, \mathrm{d}x} \geq \inf_{x\geq 0} \frac{1- \Prob(X_1 \leq x)^{pn(1-\varepsilon)}}{1- \Prob(X_1 \leq x)^{pn(1+\varepsilon)}}  \geq \inf_{v\in [0,1)} f(v)
		\end{align*}
		where $f(v)= {(1-v^{1-\varepsilon})}/{(1-v^{1+\varepsilon})}$. Now the conclusion follows by using the fact that the function $f$ is nonincreasing and that $\inf_{v\in [0,1)}f(v)=\lim_{v\to 1} f(v)= {(1-\varepsilon)}/{(1+\varepsilon)}\geq (1-\varepsilon)^2$.
	\end{proof}

	Putting together these two claims, we obtain an algorithm that checks at most $(1+\varepsilon)pn$ items and guarantees
	\[
	\gamma (1-\varepsilon)^2 (1-\delta)\E\left[ \max_{i\leq (1+\varepsilon)pn} X_i  \right] \leq \E\left[ \mathrm{Val}(\mathcal{B}')  \right] +\delta.
	\]


	Now, fix $\varepsilon>0$ small enough such that $(1+\varepsilon)p < 1$. We know that the set $\{ \lfloor (1+\varepsilon)pn \rfloor  \}_{n\geq 1}$ contains all non-negative integers. Thus, for $n \geq {2}\log \left( {2}/{\delta} \right)/({p\varepsilon^{2}})$ algorithm $\mathcal{B}'$ in an input of length $m=\lfloor (1+\varepsilon)pn \rfloor$ guarantees
	\[
	\gamma (1-\varepsilon)^2 (1-\delta) \E\left[ \max_{i\leq m} X_i  \right] \leq \E\left[ \mathrm{Val}(\mathcal{B}')  \right] +\delta.
	\]
	for any distribution $\mathcal{D}$ with support in $[0,1]$. This finishes the proof.
\end{proof}

The next result uses notation from the work by Hill \& Kertz. For the details we refer the reader to the work~\citep{hill1982comparisons}. The result states that there is a hard instance for the i.i.d.\ prophet inequality problem where $\E[\max_{i\leq m} X_i]$ is away from $0$ by a quantity at least $1/m^3$. The importance of this reformulation of the result by Hill \& Kertz is that $e^{-\Theta(n)}/\E[\max_{i\leq m} X_i] \to 0$ which is what we needed to show that $\gamma \leq 1/\beta$. Recall that $\beta\approx 1.341$ is the unique solution of the integral equation $\int_0^1 ( y(1-\log y) +\beta -1 )^{-1} \mathrm{d}y = 1$~\citep{kertz1986stop}.

\begin{proposition}[Reformulation of Proposition 4.4 by~\cite{hill1982comparisons}]\label{prop:hill_kertz_reform}
	Let $a_m$ be the sequence constructed by Hill \& Kertz, i.e, such that $a_m\to \beta$ and for any sequence of i.i.d.\ random variables $X_1,\ldots,X_m$ with support in $[0,1]$ we have
	\[
	\E\left[ \max_{i\leq m} X_i  \right] \leq a_m \sup\left\{ \E[X_t] : t\in T_m  \right\},
	\]
	where $T_m$ is the set of stopping rules for $X_1,\ldots,X_m$. Then, for $m$ large enough, there is a sequence of i.i.d.\ random variables $\widehat{X}_1,\ldots,\widehat{X}_m$ with support in $[0,1]$ such that
	\begin{itemize}
		\item $\E\left[ \max_{i\leq m} \widehat{X}_i  \right] \geq 1/m^3$, and 
		
		\item $\E\left[ \max_{i\leq m} \widehat{X}_i  \right] \geq(a_m- 1/m^3) \sup\left\{ \E[\widehat{X}_t] : t\in T_m  \right\}$.
	\end{itemize}
\end{proposition}

\begin{proof}
	In Proposition 4.4~\citep{hill1982comparisons}, it is shown that that for any $\varepsilon'$ sufficiently small, there is a random variable $\widehat{X}$ with $\widehat{p}_0= \Prob(\widehat{X}=0)$, $\widehat{p}_j=\Prob(\widehat{X}= V_j(\widehat{X}))$ for $j=0,\ldots,m-2$, $\Prob(\widehat{X}= V_{m-1}(\widehat{X})) = \widehat{p}_{m-1}-\varepsilon'$ and $\Prob(\widehat{X}=1)=\varepsilon'$ such that $\E[\max_{i\leq m} \widehat{X}_i] \geq (a_m - \varepsilon') \sup\left\{ \E[\widehat{X}_t] : t\in \widehat{T}_m \right\}$, where $\widehat{X}_1,\ldots,\widehat{X}_m$ are $m$ independent copies of $\widehat{X}$. Here $V_j(\widehat{X})= \E[ \widehat{X} \wedge \E[V_{j-1}(\widehat{X})] ]$ corresponds to the optimal value computed via dynamic programming and one can show that $\sup\left\{ \E[\widehat{X}_t] : t\in \widehat{T}_m \right\}= V_m(\widehat{X})$ (see Lemma 2.1 in~\citep{hill1982comparisons}). We only need to show that we can choose $\varepsilon'=1/m^3$. The probabilities $\widehat{p}_0,\ldots,\widehat{p}_{m-1}$ are computed as follows: Let $\widehat{s}_j = (\eta_{j,m}(\alpha_m))^{1/m}$ for $j=1,\ldots,n-2$ where $\alpha_m\in(0,1)$ is the (unique) solution of $\eta_{m-1,m}(\alpha_m)=1$, then $\widehat{p}_0=\widehat{s}_0$, $\widehat{p}_j= \widehat{s}_j-\widehat{s}_{j-1}$ for $j=1,\ldots,n-2$ and $\widehat{p}_{n-1}= 1 - \widehat{s}_{n-2}$. One can show that $\widehat{s}_{m-2} = \left( 1-1/m \right)^{1/(m-1)}\left( 1 - { \alpha_m}/{m}  \right)^{1/(m-1)}$ and $\alpha_m $ holds $ 1/(3e) \leq \alpha_m \leq 1/(e-1)$ (see Proposition 3.6 in~\citep{hill1982comparisons}). For $m$ large we have 
	\[
	e^{-1/(m-1)} \leq \widehat{s}_{m-2} \leq e^{-1/(3em^2)}
	\]
	then $\widehat{p}_{m-1}= 1- \widehat{s}_{m-2} \geq 1 - e^{-1/(m-1)} \geq 1/m^2$ for $m$ large. Thus we can set $\varepsilon'=1/m^3$ and $\widehat{p}_{m-1}-\varepsilon' >0$ and the rest of the proof follows. Furthermore, $\E[\max_{i\leq m} X_i] \geq \varepsilon' \cdot 1 \geq 1/m^3$.
\end{proof}

\subsection{Missing proofs from Section~\ref{sec:exact_sol_large_p}}\label{subsec:app:exact_sol_large_p}

\begin{proof}[Proof of Lemma~\ref{lem:ratio_bound}] For $p\geq p^*$ and $\ell=0,1,\ldots, 4$, we calculate tight lower bounds for the expression in the letf-hand side of the inequality in the claim, and we show that these lower bounds are at least one, with the lower bound attaining equality with $1$ for $\ell=1,2$. For $\ell \geq 5$ we can generalize the previous bounds and show a universal lower bound of at least $1$.
	\begin{itemize}[leftmargin=*]
		\item For $\ell=0$, we have
		\begin{align*}
			\int_{p^{{1}/{(1-p)}}}^1 \frac{1}{t^p} \,\mathrm{d}t & = \frac{1}{1-p} \left( 1 -p \right) = 1 = (1-p)^0 .
		\end{align*}
		\item For $\ell=1$, we have
		\begin{align*}
			\int_{p^{{1}/{(1-p)}}}^1 \frac{(1-t)}{t^p} \,\mathrm{d}t & = 1 - \int_{p^{{1}/{(1-p)}}}^1 t^{1-p} \,\mathrm{d}t  = 1 - \frac{1}{2-p} \left( 1- p^{{(2-p)}/{(1-p)}}  \right) .
		\end{align*}
		The last value is at least $1-p$ if an only if $p (2-p) \geq 1-p^{{(2-p)}/{(1-p)}}$ iff $p^{(2-p)/(1-p)}\geq (1-p)^2$. The last inequality holds iff $p\geq p^* \approx 0.594134$ where $p^*$ is computed numerically by solving $(1-p)^2 = p^{(2-p)/(1-p)}$.
		
		\item For $\ell = 2$, we use the approximation $p^{1/(1-p)} \leq {(1+p)}/{(2e)}$ that follows from the concavity of the function $p^{1/(1-p)}$ and the first-order approximation of the function at $p=1$. With this we can lower bound the integral
		\begin{align*}
			\int_{p^{1/(1-p)}}^1 \frac{(1-t)^2}{t^p} \, \mathrm{d}t & \geq \int_{{(1+p)}/{(2e)}}^1 \frac{(1-t)^2}{t^p} \, \mathrm{d}t \\
			& = \int_0^{1-{(1+p)}/{(2e)}} u^2 (1-u)^{-p} \, \mathrm{d}u \tag{change of variable $u=1-t$} \\
			& \geq \int_0^{1-{(1+p)}/{(2e)}} u^2 \left( 1 + p u + p(p+1) \frac{u^2}{2} \right) \, \mathrm{d} u \tag{Using the series $(1-u)^{-p} = \sum_{k\geq 0} \binom{-p}{k} (-u)^k$ } \\
			& =  \frac{1}{3}\left(1-\frac{1+p}{2e}\right)^3 + \frac{p}{4}\left(1-\frac{1+p}{2e}\right)^4 + \frac{p(p+1)}{10} \left( 1- \frac{1+p}{2e} \right)^5.
		\end{align*}
		By solving the polynomial we see that the last expression is $\geq(1-p)^2$ if and only if $p \geq 0.585395 $, thus the inequality holds for $p\geq p^*$.
		
		\item For $\ell=3,4$ we can use a similar approach to get
		\[
		\int_{p^{{1}/{(1-p)}}}^1 \frac{(1-t)^\ell}{t^p} \, \mathrm{d}t \geq  \frac{1}{\ell+1} \left( 1-\frac{1+p}{2e} \right)^{\ell+1} + \frac{p}{\ell+2} \left( 1 - \frac{1+p}{2e}  \right)^{\ell+2}.
		\]
		The last expression is $\geq (1-p)^{\ell}$ for $\ell=3,4$ if and only if $p\geq 0.559826$.
		
		\item For $\ell \geq 5$, we have
		\begin{align*}
			\int_{p^{{1}/{(1-p)}}}^1 \frac{(1-t)^\ell}{t^p} \, \mathrm{d} t \geq \frac{(1-{(1+p)}/{(2e)})^{\ell+1}}{\ell+1} .
		\end{align*}
		We show that ${\left( 1 - {(1+p)}/{(2e)}  \right)^{\ell+1}}/{(\ell+1)} \geq (1-p)^\ell$. This is equivalent to
		$$\left(\frac{1-(1+p)/(2e)}{1-p}\right)^\ell \left( 1- \frac{1+p}{2e} \right)\geq \ell+1.$$
		Note that the function $f(p) = {\left(1-(1+p)/(2e)\right)}/{(1-p)}$ is increasing since $f'(p) = {(1-1/e)}/{(1-p)^2} > 0$. For $\overline{p} = {(2e-1)}/{(4e-3)} \approx 0.56351$ we have $f(\overline{p}) = 2-1/e $. Thus for $p\geq p^* > \overline{p}$ and $\ell =5$ we have $f(p)^5\left( 1 - {(1+p)}/{(2e)} \right) \geq (2-1/e)^5 (1-1/e)\geq 7.32 \geq 6$. By an inductive argument, we can show that $f(p)^\ell \left( 1 - {(1+p)}/{(2e)}  \right) \geq \ell+1$ for any $\ell\geq 5$ and this finishes the proof. 
	\end{itemize}
	
\end{proof}

\begin{proof}[Proof of Lemma~\ref{lem:limit_log_tk}]
	During the proof, we assume that $1/p \notin \N$. This is an assumption that is easy to remove with a density argument. We divide the proof into a series of propositions and lemmas.
	
	Taking logarithm on both sides of Identity~\eqref{eq:definition_of_tk} we obtain
	\[
	\log t_{k+1} - \log t_k = \frac{1}{1-kp} \log \left(  \frac{A_k(1-p)}{A_{k-1}} \right).
	\]
	From here, we obtain
	\begin{align*}
		\log t_{\lfloor 1/p \rfloor + 1} - \log t_2 &= \sum_{j=2}^{\lfloor 1/p \rfloor} \frac{1}{1-jp} \log\left( \frac{A_j(1-p)}{A_{j-1}} \right) \\
		& = \sum_{j=2}^{\lfloor 1/p \rfloor} \frac{1}{1-jp}\int_{A_{j-1}}^{(1-p)A_j} \frac{1}{x} \, \mathrm{d}x.,
	\end{align*}
	and also
	\begin{align*}
		\log t_{k+1} - \log t_{\lfloor 1/p \rfloor + 1} &= \sum_{j=\lfloor 1/p \rfloor}^k \frac{1}{jp-1} \log \left( \frac{A_{j-1}}{A_j(1-p)}  \right) \\
		& = \sum_{j=\lfloor 1/p \rfloor}^k \frac{1}{jp-1} \int_{A_j(1-p)}^{A_j} \frac{1}{x} \, \mathrm{d}x.
	\end{align*}
	
	\begin{proposition}
		We have
		\begin{enumerate}
			\item For $k <1/p$,
			\[
			\frac{\gamma p (1-kp)}{A_k(1-p)} \leq  \int_{A_{k-1}}^{A_k (1-p)} \frac{1}{x} \, \mathrm{d}x \leq \frac{\gamma p (1-kp)}{A_{k-1}}
			\]
			
			\item For $k> 1/p$,
			\[
			\frac{\gamma p (kp-1)}{A_{k-1}} \leq  \int_{A_{k}(1-p)}^{A_{k-1}} \frac{1}{x} \, \mathrm{d}x \leq \frac{\gamma p (kp - 1)}{A_{k}(1-p)}
			\]
		\end{enumerate}
	\end{proposition}
	
	\begin{proof}
		Both results follow by using the monotonicity of $1/x$ and that
		\[
		A_k(1-p) - A_{k-1}  = t_1(1-p)^{-k+1} + \gamma p k (1-p) - t_1(1-p)^{-k+1} - \gamma p (k-1) \\
		 = \gamma p (1-kp).
		\]
	\end{proof}
	
	The next result shows bound over $\log t_{k+1}$. We use this result to interpret the bounds as Riemann sums.
	
	\begin{proposition}[Bounds on $\log t_{k+1}$]
		For $k\geq 1/p$, we have
		\[
		\sum_{j=2}^{k-1} \frac{\gamma p}{A_j}\leq \log t_{k+1} -\log t_2 \leq \sum_{j=1}^{k} \frac{\gamma p}{A_j} + \frac{p}{1-p} \sum_{j=\lfloor 1/p \rfloor+1}^{k} \frac{\gamma p}{A_j}.
		\]
	\end{proposition}
	
	\begin{proof}
		For the upper bound we have
		\begin{align*}
			\log t_{k+1} &\leq \log t_2 + \sum_{j=2}^{\lfloor 1/ p \rfloor} \frac{\gamma p}{A_{j-1}} + \frac{1}{1-p} \sum_{j=\lfloor 1/p \rfloor+1}^k \frac{\gamma p}{A_j} \\
			& \leq \log t_2 + \sum_{j=1}^{k} \frac{\gamma p}{A_j} + \frac{p}{1-p}\sum_{j=\lfloor 1/p \rfloor+1}^k \frac{\gamma p}{A_j}
		\end{align*}
		For the lower bound we have
		\begin{align*}
			\log t_{k+1} & \geq \log t_2 +  \frac{1}{1-p}\sum_{j=2}^{\lfloor 1/p \rfloor} \frac{\gamma p}{A_{j}} + \sum_{j=\lfloor 1/p \rfloor + 1}^{k} \frac{\gamma p}{A_{j-1}} \geq \log t_2 + \sum_{j=2}^{k-1} \frac{\gamma p}{A_j}.
		\end{align*}
	\end{proof}

	For $p>0$ but small enough, $t_1 e^{jp} + \gamma p j \leq A_j \leq t_1 e^{jp/(1-p)} + \gamma p j$. Using this in the bounds of the previous proposition, we obtain
	\[
	\int_2^{\infty} \frac{\gamma p}{t_1 e^{xp/(1-p) + \gamma xp }}\, \mathrm{d}x  \leq \lim_{k\to \infty}  \log t_{k+1}  -\log t_2 \leq \frac{\gamma p}{A_1} + \int_1^\infty \frac{\gamma p}{t_1 e^{xp} + \gamma xp} \, \mathrm{d}x + \frac{p}{1-p} \int_{\lfloor 1/p \rfloor}^\infty \frac{\gamma p}{t_1 e^{px} +\gamma p x} \, \mathrm{d}x.
	\]
	Note that 
	\[
	\frac{p}{1-p}\int_{\lfloor 1/p \rfloor }^k \frac{\gamma p}{ t_ 1 e^{px} + \gamma p x}\, \mathrm{d}x \leq \frac{p}{1-p} \int_{1}^\infty \frac{\gamma}{t_1} e^{-x}\,\mathrm{d}x =  \frac{p}{1-p}\frac{\gamma}{t_1} e^{-1}.
	\]
	Then, taking $p\to 0$, we obtain
	\[
	\int_0^\infty \frac{\gamma}{t_1e^{x} + \gamma x} \, \mathrm{d}x = \lim_{k\to \infty} \log t_{k} - \log t_1
	\]
	where we used that $t_2= t_1 \left( 1+ \gamma p(1-p)/t_1 \right)^{1/(1-p)}\to t_1$ when $p\to 0$. This concludes the proof of Lemma~\ref{lem:limit_log_tk}.
\end{proof}

\end{document}